%% file: main.tex
\documentclass[11pt]{article}

\usepackage{hdb_macros}
\usepackage{fullpage}
\usepackage{thm-restate}

\usepackage[multiuser,inline,nomargin]{fixme}
\fxusetheme{color}
\FXRegisterAuthor{h}{eh}{\color{blue}Huck}
\FXRegisterAuthor{k}{ek}{\color{red}Kyle}

\newcommand{\basis}{B}
\renewcommand{\O}{\ensuremath{\mathsf{O}}}

\newcommand{\SDLIP}{\problem{SDLIP}}

\newcommand{\SearchSDPCE}{\problem{Search}\text{-}\problem{SDPCE}}
\newcommand{\LIP}{\problem{LIP}}

\newcommand{\ZLIP}{\Z\problem{LIP}}
\newcommand{\SPD}{S^{>0}}

\newcommand{\PCE}{\problem{PCE}}
\newcommand{\GI}{\problem{GI}}

\renewcommand{\C}{\mathcal{C}}
\newcommand{\cP}{\mathcal{P}}
\newcommand{\permequiv}{\stackrel{P}{\approx}}
\DeclareMathOperator{\wt}{wt}

\DeclareMathOperator{\clen}{len}

\DeclareMathOperator{\supp}{supp}

\renewcommand{\NL}{N_L}
\renewcommand{\NC}{N_C}
\newcommand{\cZ}{\mathcal{Z}}
\DeclareMathOperator{\adj}{adj}
\renewcommand{\len}{\mathrm{len}}

\begin{document}

\title{On (Non-)Isomorphism of Self-Dual Lattices and Codes}
\author{Huck Bennett\thanks{University of Colorado Boulder. \email{huck.bennett@colorado.edu}. Supported in part by NSF Award No. 2432132.} \and Kyle Fridberg\thanks{Cornell University. \email{kof4@cornell.edu}. Supported by an NSF Graduate
Research Fellowship under Grant No. DGE–2139899}}
\date{\today}

\maketitle

\input{abstract}

\newpage

\input{introduction}

\input{prelims}
\input{SDLIP}

\input{SDPCE}

\bibliography{unimod-iso}
\bibliographystyle{alpha}

\end{document}

%% file: abstract.tex
\begin{abstract}
A recent line of work motivated by cryptographic applications has studied the complexity of the \emph{Lattice Isomorphism Problem} ($\LIP$).
In this work, we study $\LIP$ on \emph{self-dual} lattices $\lat \subset \R^n$, which appear naturally in many applications.
Our main results are a $2^{n/2 + o(n)}$-time randomized algorithm for $\LIP$ and a $\coNP$ protocol for $\LIP$ on a broad class of self-dual lattices.
These results extend recent work on $\ZLIP$, the problem of deciding whether a lattice is isomorphic to $\Z^n$.
In particular, the former result extends the $2^{n/2 + o(n)}$-time algorithms for $\ZLIP$ of Bennett, Ganju, Peetathawachai, and Stephens-Davidowitz (Eurocrypt, 2023) and of Ducas (Des. Codes Cryptogr., 2024). The latter result extends the $\ZLIP \in \coNP$ result of Hunkenschr\"{o}der (Math. Prog. Series A, 2024).

Our results leverage two key structural properties of self-dual lattices $\lat \subset \R^n$: (1) every such lattice $\lat$ is isomorphic to $\lat_0 \oplus \Z^r$ for some self-dual lattice $\lat_0$ with $\lambda_1(\lat_0)^2 \geq 2$, and (2) every such lattice $\lat$ has \emph{characteristic vectors}, i.e., there exist vectors $\vec{w} \in \lat$ such that for every $\vec{v} \in \lat$, $\iprod{\vec{v}, \vec{w}} \equiv \iprod{\vec{v}, \vec{v}} \Mod{2}$.
Our results use a line of work by Elkies and Gaulter on lattices with long shortest characteristic vectors, and can be strengthened assuming a positive answer to a related question of Elkies (Math. Res. Lett., 1995).

We also study Permutation Code Equivalence ($\PCE$) on self-dual codes, and we observe that similar structural properties imply a polynomial-time algorithm for $\PCE$ on certain such codes.
This gives a natural class of codes with large hull for which $\PCE$ is easy.
\end{abstract}

%% file: introduction.tex
\section{Introduction}

A lattice $\lat$ is a discrete, additive subgroup of $\R^n$.
The decisional \emph{Lattice Isomorphism Problem} (LIP) asks, roughly speaking, whether two lattices $\lat, \lat' \subset \R^n$ are rotations of one another.
$\LIP$ has received substantial attention from an algorithmic and complexity-theoretic perspective (see, e.g.,~\cite{szydloHypercubicLatticeReduction2003,havivLatticeIsomorphismProblem2014,lenstraLatticesSymmetry2017,chandrasekaranDecidingOrthogonalityConstructionA2017}).
Moreover, a thriving line of recent work initiated by~\cite{conf/eurocrypt/DucasW22,conf/eurocrypt/BennettGPS23} has developed cryptosystems that are inspired by $\LIP$ and whose security relies on the hardness of $\LIP$.
This work has been quite fruitful already. In particular, it has led to the development of advanced cryptographic primitives based on LIP (e.g.,~\cite{branco25FHELIP,leporati2025beyond,BLS26-IBE-FHE-crypto}) and to the practical digital signature algorithm HAWK~\cite{conf/asiacrypt/DucasPPW22}.\footnote{HAWK is the sole lattice-based cryptosystem to appear as a second-round candidate in the ``Additional Digital Signature Schemes'' standardization process run by the National Institute of Standards and Technology (NIST)~\cite{nist_pqc_round2_additional_signatures_2025}.}
Many of these works concern $\LIP$ on certain nice classes of lattices such as module lattices in~\cite{conf/asiacrypt/DucasPPW22} and the integer lattice $\Z^n$ in~\cite{conf/eurocrypt/DucasW22,conf/eurocrypt/BennettGPS23}.

The main goal of this work is to better understand the complexity of $\LIP$ on \emph{self-dual lattices}, i.e., lattices $\lat \subset \R^n$ such that $\lat = \lat^*$, where $\lat^* := \set{\vec{x} \in \R^n: \forall \vec{y} \in \lat, \iprod{\vec{x}, \vec{y}} \in \Z}$ is the dual lattice of $\lat$.\footnote{Self-dual lattices are also called \emph{unimodular lattices} since Gram matrices of their bases are necessarily unimodular matrices. In this work, we choose to call these lattices self-dual since we also work with self-dual codes. Throughout this work we assume by default that lattices are full-rank, i.e., we assume that if $\lat \subset \R^n$ then $\lspan(\lat) = \R^n$.}
Self-dual lattices are heavily studied in number theory \cite{Conway-Sloane-SPLAG}, and they appear in topology (e.g.,~\cite{Freedman-four-manifolds-82, Donaldson1983}) and sphere packing.
In particular, optimal sphere packings in $8$ and $24$ dimensions are achieved by the famous self-dual lattices $E_8$ and $\Lambda_{24}$, respectively~\cite{Viazovska2017-SpherePacking-E8,Cohn2017-SpherePacking-Lambda24}.
Moreover, the list of LIP challenges given in~\cite[Table 1]{conf/eurocrypt/DucasW22} includes several self-dual lattices.
The integer lattice $\Z^n$ is a self-dual lattice, and so $\ZLIP$, the problem of deciding whether a given input lattice is isomorphic to $\Z^n$, is a very special case of $\LIP$ on self-dual lattices.

The fastest known algorithm for (general) $\LIP$ runs in $n^{O(n)}$ time~\cite{havivLatticeIsomorphismProblem2014}.
Faster algorithms are known on so-called Construction-A lattices~\cite{chandrasekaranDecidingOrthogonalityConstructionA2017,conf/pkc/DucasG23} and on rank-$2$ module lattices over certain number fields~\cite{conf/eurocrypt/MureauPPW24,conf/eurocrypt/AllombertPW25}.
Furthermore, the fastest known algorithms for $\ZLIP$ run in $2^{n/2 + o(n)}$ time~\cite{conf/eurocrypt/BennettGPS23,journals/dcc/Ducas24}, and to the best of our knowledge very little is known about $\LIP$ on broader classes of self-dual lattices than rotations of $\Z^n$.

We also consider the complexity of $\LIP$ more broadly.
Intuitively, it is straightforward to certify that two lattices $\lat, \lat'$ \emph{are} isomorphic by providing an orthogonal transformation $O$ such that $O(\lat) = \lat'$ as a certificate. This essentially shows that $\LIP \in \NP$.\footnote{In fact, we will formalize $\LIP$ in terms of quadratic forms rather than lattices in the sequel, in which case the witness is a bit different.}
However, it is a priori unclear how to certify that two lattices are not isomorphic, and it is an open question whether $\LIP \in \coNP$. 

In fact, showing that $\LIP$ is in $\coNP$ would imply that the Graph Isomorphism Problem ($\GI$) is in $\coNP$ by a known polynomial-time reduction from $\GI$ to $\LIP$~\cite{journals/moc/SikiricSV09}, thus resolving another open problem. (This reduction also shows that $\LIP$ is $\GI$-hard, i.e., that $\LIP$ has a polynomial-time algorithm only if $\GI$ does.)
Additionally,~\cite{havivLatticeIsomorphismProblem2014} showed that $\LIP$ is in $\SZK$ and therefore that it is in $\coAM$, i.e., that there is an Arthur-Merlin protocol for certifying that two lattices are not isomorphic. This implies that $\LIP$ is unlikely to be $\NP$-complete since then the polynomial hierarchy would collapse.
Recently, Hunkenschr\"{o}der~\cite{Hunkenschroder-ZLIP-coNP} used a result of Elkies~\cite{Elkies-char-Zn} to give a simple, elegant proof that $\ZLIP \in \coNP$. To the best of our knowledge, nothing else is known about the complexity of $\LIP$ on self-dual lattices.

We also study the analog of $\LIP$ on self-dual codes.
A (binary, linear) \emph{code} $\C$ is a linear subspace of $\F_2^n$.
As with lattices, the question of whether two codes are ``essentially the same''---formalized in terms of variants of the \emph{Code Equivalence Problem}---has garnered substantial attention in the cryptographic literature, including as the basis of the LESS digital signature scheme~\cite{conf/africacrypt/BiasseMPS20,conf/pqcrypto/BarenghiBPS21}.\footnote{Like HAWK, LESS is a second-round candidate in NIST's ``Additional Digital Signature Schemes'' standardization process~\cite{nist_pqc_round2_additional_signatures_2025}. LESS uses codes over large finite fields $\F_q$. Specifically, the parameter sets submitted to~\cite{nist_pqc_round2_additional_signatures_2025} take $q = 127$. Its security also relies on the hardness of the \emph{Linear Code Equivalence Problem}, which consider a more general notion of equivalence allowing for both permuting and scaling coordinates.}
In particular, the \emph{Permutation} Code Equivalence Problem ($\PCE$) asks whether two codes $\C, \C' \subseteq \F_2^n$ are the same up to a permutation of coordinates. 
A code $\C \subset \F_2^n$ is \emph{self-dual} if $\C = \C^{\perp}$, where $\C^{\perp} := \set{\vec{x} \in \F_2^n : \forall \vec{y} \in \C, \iprod{\vec{x}, \vec{y}} = 0}$ is the dual code of $\C$.\footnote{Here and in general we emphasize that all arithmetic over codes in $\F_2$ is mod $2$.}

\subsection{Reduced Lattices and Characteristic Vectors}
\label{sec:intro-char-vecs}

We write $\lat \cong \lat'$ if lattices $\lat, \lat'$ are isomorphic (i.e., if there exists an orthogonal transformation $O$ such that $O(\lat) = \lat'$). We call a self-dual lattice $\lat_0$ a \emph{reduced lattice} if its squared minimum distance is at least $2$ (i.e., if $\lambda_1(\lat_0)^2 \geq 2$).
To describe our results, we first introduce and discuss two key properties of self-dual lattices, which are as follows.

\begin{enumerate}
\item \label{item:intro-reduced} Every self-dual lattice $\lat$ is isomorphic to $\lat_0 \oplus \Z^r$ for a unique (up to isomorphism) reduced lattice $\lat_0$ and some $r \geq 0$.\footnote{We note that either $\lat_0$ or $\Z^r$ may be trivial. I.e., it may be the case that $\lat = \lat_0$ or $\lat = \Z^r$.} 
We define $R_0(\lat) := \rank(\lat_0)$, and we call a reduced lattice $\lat_0$ such that $\lat \cong \lat_0 \oplus \Z^r$ a \emph{reduction} of $\lat$.

\item \label{item:intro-characteristic-vec} Every self-dual lattice $\lat$ has \emph{characteristic vectors}, i.e., there exist vectors $\vec{w} \in \lat$ such that for every $\vec{v} \in \lat$, $\iprod{\vec{v}, \vec{w}} \equiv \iprod{\vec{v}, \vec{v}} \Mod{2}$. Moreover, $\vec{w}'$ is a characteristic vector of $\lat$ if and only if $\vec{w}' \in 2 \lat + \vec{w}$ for some fixed characteristic vector $\vec{w}$. Following the notation of Gaulter~\cite{Gaulter-S2}, we define $s(\lat) := \min \set{\norm{\vec{w}'}^2 : \vec{w}' \in 2 \lat + \vec{w}}$ to be the minimum squared Euclidean norm of a characteristic vector of $\lat$.
\end{enumerate}

The coset $\lat + \half \vec{w}$ of $\lat$ consisting of all $\half$-scaled characteristic vectors of $\lat$ is called the \emph{shadow} of $\lat$ in the literature 
(see, e.g.,~\cite[p. xxxiv]{Conway-Sloane-SPLAG} and \cite{Elkies-long-shadows}). We note that the parameters $R_0$ and $s$ are both invariant under isomorphism.
Analogous properties also hold for self-dual codes $\C \subset \F_2^n$, although in the introduction we focus on lattices.
(See \cref{sec:prelims} for formal statements and corresponding references or proofs for all of the facts claimed here.)

We will crucially use a line of work, initiated by Elkies~\cite{Elkies-char-Zn} and continuing with works by Elkies and Gaulter~\cite{Elkies-long-shadows,Gaulter-Lattices-Without-Short,Gaulter-S2}, on characterizing which self-dual lattices $\lat$ have large $s(\lat)$. 
One can show that if $\lat \subset \R^n$ is a self-dual lattice then $s(\lat) = n - 8K$ for some $K \in \Z_{\geq 0}$ (see \cref{fact:char-vec-properties}, \cref{item:char-vec-congruences}), and in fact the work of Elkies and Gaulter is framed in terms of characterizing the self-dual lattices $\lat \subset \R^n$ with $s(\lat) = n - 8K$ for small values of $K \geq 0$. 
(\cite{Elkies-long-shadows} also studies the analogous question for codes.) 

Starting with $K = 0$,~\cite{Elkies-char-Zn} showed that $s(\lat) = n$ if and only if $\lat \cong \Z^n$.
This led to Hunkenschr\"{o}der's~\cite{Hunkenschroder-ZLIP-coNP} $\coNP$ protocol for $\ZLIP$.
He noted that to certify that a self-dual lattice $\lat \subset \R^n$ is not isomorphic to $\Z^n$, it suffices to provide a characteristic vector $\vec{w}$ of $\lat$ with $\norm{\vec{w}}^2 < n$ as a witness. (Although we are continuing to present our overview in terms of lattices,~\cite{Hunkenschroder-ZLIP-coNP} like us formally works with quadratic forms $Q = \basis^T \basis$ equal to the Gram matrix of a basis $\basis$ of $\lat$. Accordingly, the $\coNP$ protocol in~\cite{Hunkenschroder-ZLIP-coNP} in fact provides the \emph{coefficient vector} $\vec{z}$ of such a characteristic vector $\vec{w} = \basis \vec{z}$. It checks that $\basis \vec{z}$ is a characteristic vector of $\lat$ and that $\vec{z}^T Q \vec{z} < n$.)

 Following~\cite{Elkies-long-shadows}, we define the quantity $\NL(K) := \sup \set{R_0(\lat) : \lat = \lat^*, s(\lat) \geq \rank(\lat) - 8K}$.\footnote{We note that if $\lat \cong \lat_0 \oplus \Z^r$ then $s(\lat) = s(\lat_0) + r$ 
(see \cref{fct:properties-s}, \cref{item:properties-s-direct-sum}), and so equivalently one can take the supremum over reduced lattices $\lat_0$ as opposed to arbitrary self-dual lattices $\lat$.}
So, by definition, if $s(\lat) \geq n - 8K$ then $R_0(\lat) \leq \NL(K)$.

In~\cite{Elkies-long-shadows}, Elkies showed that $\NL(1) = 23$ and asked whether $\NL(K)$ is finite for every $K$. %
(In fact,~\cite{Elkies-long-shadows} showed that, up to isomorphism, there are precisely $14$ reduced lattices $\lat_0 \subset \R^n$ with $s(\lat_0) = n - 8$. We include a list of these lattices in \cref{tab:n-8}.)
Gaulter~\cite{Gaulter-Lattices-Without-Short} subsequently gave a positive answer to this question for $K = 2$ and $K = 3$ and even gave explicit upper bounds on $\NL(2)$ and $\NL(3)$. Gaulter also gave an improved upper bound on $\NL(2)$ in~\cite{Gaulter-S2}.
These bounds are important for our work, and we summarize them in \cref{thm:Nkbounds}.
To the best of our knowledge, nothing is known in general for $K \geq 4$, although there has been substantial progress in bounding the rank of reduced lattices $\lat_0$ that simultaneously have large $\lambda_1(\lat_0)$ and large $s(\lat_0)$ \cite{gaulter1998characteristic, NebeVenkov2003, Gaborit2007, NebeSchindelar2007}.

\subsection{Results and Technical Overview}
\label{sec:intro-our-results}

We now state our results. To do this formally for lattices, we use quadratic forms. A \emph{quadratic form} $Q$ corresponding to a lattice $\lat \subset \R^n$ is the Gram matrix $Q := \basis^T \basis$ of a basis $\basis$ of $\lat$.\footnote{Formally, a quadratic form is the scalar-valued function $\vec{z} \mapsto \vec{z}^T Q \vec{z}$. For simplicity, we directly refer to the matrix $Q$ as a quadratic form.}
We write $\lat \cong \lat'$ if two lattices $\lat, \lat' \subset \R^n$ are isomorphic (i.e., if there exists an orthogonal transformation $O$ such that $O(\lat) = \lat'$).

\paragraph{An algorithm for $\LIP$ on self-dual lattices.}

Our first result is an algorithm for search $\LIP$ on self-dual lattices $\lat, \lat'$ (i.e., when $\lat \cong \lat'$, and the goal is to \emph{find} an isomorphism between the lattices) with relatively low-rank reductions $\lat_0, \lat_0'$.

\begin{restatable}{theorem}{lowRlatalg} \label{thm:intro-low-R0-lattice-alg}
There is an algorithm for search $\LIP$ on quadratic forms $Q, Q'$ corresponding to self-dual lattices $\lat, \lat' \subset \R^n$ that runs in $2^{n/2 + o(n)} + n_0^{O(n_0)}$ time, where $n_0 := R_0(\lat) = R_0(\lat')$.
In particular, if $n_0 = o(n/\log n)$ then the algorithm runs in $2^{n/2 + o(n)}$ time.
\end{restatable}

We emphasize again that the fastest known algorithm for $\LIP$ on general lattices takes $n^{O(n)}$ time~\cite{havivLatticeIsomorphismProblem2014} (see also \cref{thm:HR-LIP}), and that \cref{thm:intro-low-R0-lattice-alg} with $n_0 = o(n/\log n)$ matches the running time of the algorithms for $\ZLIP$ in~\cite{conf/eurocrypt/BennettGPS23,journals/dcc/Ducas24}, but applies to a much larger class of self-dual lattices (note that $R_0(\lat) = 0$ for $\lat \cong \Z^n$).

The algorithm works by reducing $\LIP$ on $\lat, \lat'$ to $\LIP$ on their reductions $\lat_0, \lat_0'$, and then using the algorithm in~\cite{havivLatticeIsomorphismProblem2014} to solve $\LIP$ on $\lat_0, \lat_0'$.\footnote{We note the unfortunate clash in terminology between a reduction in the standard algorithmic sense and a reduction $\lat_0$ of a self-dual lattice $\lat$.}
The reduction works by finding all unit-length vectors in $\lat \cong \lat_0 \oplus \Z^r, \lat' \cong \lat_0' \oplus \Z^r$---i.e., by finding the sublattices of $\lat, \lat'$ isomorphic to $\Z^r$---and then projecting orthogonally to obtain $\lat_0, \lat_0'$. 
To implement this reduction, we show that the algorithm from~\cite{conf/eurocrypt/BennettGPS23} for finding unit-length vectors in rotations of $\Z^n$ also works for finding unit-length vectors in arbitrary self-dual lattices (if any exist).

By definition of $\NL(K)$, we get that if $\NL(K) = o(n/\log n)$ for $K = K(n)$, then \cref{thm:intro-low-R0-lattice-alg} solves $\LIP$ on lattices $\lat, \lat' \subset \R^n$ with $s(\lat), s(\lat') \geq n - 8K$ in $2^{n/2 + o(n)}$ time.
In particular, by the work of Elkies and Gaulter, \cref{thm:intro-low-R0-lattice-alg} applies when $s(\lat), s(\lat') \geq n - 8K$ for $K = 3$ unconditionally. It applies for arbitrary constant $K$ assuming a positive answer to Elkies's question about $\NL(K)$ being bounded.
Finally, we note that although in this work we focus on self-dual lattices, \cref{thm:intro-low-R0-lattice-alg} actually applies to \emph{integral lattices}---lattices $\lat$ that satisfy the condition $\lat \subseteq \lat^*$ but not necessarily the stronger condition $\lat = \lat^*$ required to be self-dual.

\paragraph{An $\NP \cap \coNP$ protocol for $\LIP$ on self-dual lattices.}

Our next result is a $\coNP$ protocol for $\LIP$ on certain self-dual lattices $\lat, \lat'$. 
That is, we give a protocol for certifying that $\lat$ and $\lat'$ are \emph{not} isomorphic when at least one of $s(\lat)$ and $s(\lat')$ is sufficiently large.
We formalize this by introducing and studying a decisional promise problem, $\SDLIP_K$ (``self-dual $\LIP$''), parameterized by an integer $K \in \Z_{\geq 0}$. See \cref{def:SDLIP}.

An instance of $\SDLIP_K$ consists of quadratic forms corresponding to two self-dual lattices $\lat, \lat' \subset \R^n$. It is a YES instance if $s(\lat) \geq n - 8K$ and $\lat \cong \lat'$ (these conditions additionally imply that $s(\lat') \geq n - 8K$), and it is a NO instance if $\lat \not\cong \lat'$ and either $s(\lat) \geq n - 8K$ or $s(\lat') \geq n - 8K$ (or both).
(Notice that the requirement that $s(\lat) \geq n - 8K$ in the YES case makes showing that $\SDLIP_K \in \NP$ non-trivial, and we prove this in addition to $\SDLIP_K \in \coNP$ for certain $K$.)

We motivate the definition of $\SDLIP_K$ by noting several useful points about it. First, $\SDLIP_K$ trivially reduces to $\SDLIP_{K+1}$ for every $K \geq 0$.\footnote{Moreover, $\SDLIP_\infty$ is simply the problem of deciding whether two self-dual lattices (represented by quadratic forms) are isomorphic.} Second, the problem of deciding whether a single input lattice $\lat$ is isomorphic to a fixed self-dual lattice $\lat'$ with $s(\lat') \geq n - 8K$ trivially reduces to $\SDLIP_K$.\footnote{More properly, one can consider a fixed family $\set{\lat_n'}$ of self-dual lattices with $\lat_n' \subset \R^n$ satisfying $s(\lat_n') \geq n - 8K$.} In particular, $\ZLIP$ reduces to $\SDLIP_0$ by fixing $\lat' := \Z^n$.
Third, YES instances $\lat, \lat'$ of $\SDLIP_K$ for $K$ such that $\NL(K) \leq n_0$ have $R_0(\lat) = R_0(\lat') \leq n_0$, and so $\SDLIP_K$ for such $K$ reduces to search $\LIP$ on lattices with low $R_0$. It is therefore solvable in $2^{n/2} + n_0^{O(n_0)}$ time using the algorithm in \cref{thm:intro-low-R0-lattice-alg}. 

Specifically, we show the following.

\begin{restatable}{theorem}{latNPcoNP} \label{thm:intro-coNP-proof}
For any $K = K(n)$ such that $\NL(K) \leq O(\log n/\log \log n)$, $\SDLIP_K$ is in $\NP \cap \coNP$.
\end{restatable}

As a corollary of~\cite{Elkies-char-Zn,Elkies-long-shadows,gaulter1998characteristic,Gaulter-S2}, we then get that $\SDLIP_3 \in \NP \cap \coNP$ unconditionally, and we get that $\SDLIP_K \in \NP \cap \coNP$ for any constant $K$ assuming a positive answer to Elkies's question about $\NL(K)$ being bounded for every $K$.

The protocol works roughly as follows. To certify either that $\lat \cong \lat'$ or $\lat \not \cong \lat'$ in the case where $s(\lat), s(\lat') \geq n - 8K$, the protocol uses a witness consisting of orthogonal transformations $O, O'$ such that $O(\lat) = \lat_0 \oplus \Z^r$ and $O'(\lat') = \lat_0' \oplus \Z^{r'}$ for some reduced, low-rank lattices $\lat_0, \lat_0'$ and some $r, r' \geq 0$. Because $\lat_0, \lat_0'$ each have rank at most $O(\log n/\log \log n)$, the verifier can check whether they are in fact reduced lattices using the $\SVP$ algorithm from~\cite{conf/stoc/AggarwalDRS15} (to check that $\lambda_1(\lat_0)^2 \geq 2$ and $\lambda_1(\lat_0')^2 \geq 2$), and can then decide whether $\lat_0 \cong \lat_0'$ using the $\LIP$ algorithm from~\cite{havivLatticeIsomorphismProblem2014} in $O(\log n/\log \log n)^{O(\log n/\log \log n)} = \poly(n)$ time.
To certify $\lat \not \cong \lat'$ when $s(\lat) < n - 8K$ but $s(\lat') \geq n - 8K$, the protocol uses a short characteristic vector $\vec{w}$ for $\lat$ together with an orthogonal transformation $O'$ such that $O'(\lat') = \lat_0' \oplus \Z^{r'}$ for a reduced lattice $\lat_0'$ as before. After verifying that $\lat_0'$ is in fact a reduced lattice, the verifier can efficiently compute $s(\lat_0')$ and therefore $s(\lat') = s(\lat_0') + r'$ (see \cref{cor:comp-shortest-char-vec,fct:properties-s}). From this, the verifier can check that $s(\lat) < s(\lat')$, and if so conclude that $\lat \not \cong \lat'$.

\paragraph{An algorithm for $\PCE$ on self-dual codes.}

We then turn to studying Permutational Code Equivalence ($\PCE$), the code analog of $\LIP$, on self-dual codes $\C \subset \F_2^n$.
Specifically, we observe that an analog of \cref{thm:intro-low-R0-lattice-alg} holds for $\PCE$.

We define a \emph{reduced code} to be a self-dual code $\C_0$ with minimum distance at least $4$, and we define $\cZ := \set{(0, 0)^T, (1, 1)^T} \subset \F_2^2$ to be the code generated by the vector $(1, 1)^T$.
The $2$-fold repetition code $\cZ^r = \cZ \oplus \cdots \oplus \cZ \subset \F_2^{2r}$ is the analog of the lattice $\Z^r$ for codes.
By analogy with \cref{item:intro-reduced} for self-dual lattices, every self-dual code $\C \subset \F_2^n$ is permutationally equivalent to $\C_0 \oplus \cZ^r$ for a unique (up to permutational equivalence) reduced code $\C_0 \subset \F_2^{n_0}$ and some $r \geq 0$. We call such a code $\C_0$ a \emph{reduction} of $\C$, and we define $R_0(\C)$ to be the block length $n_0$ of $\C_0$.
We show the following.

\begin{restatable}{theorem}{lowRcodealg} \label{thm:intro-low-R0-code-alg}
There is an algorithm for search $\PCE$ on self-dual codes $\C, \C' \subset \F_2^n$ that runs in $2^{n_0 + o(n_0)} + \poly(n)$ time, where $n_0 := R_0(\C) = R_0(\C')$. In particular, if $n_0 = O(\log n)$ then the algorithm runs in $\poly(n)$ time.
\end{restatable}

Elkies's proof in~\cite{Elkies-long-shadows} that $\C$ is permutationally equivalent to $\C_0 \oplus \cZ^r$ for some reduced code $\C_0$ is constructive, and it implicitly gives a polynomial-time algorithm for computing $\C_0$. The algorithm is essentially to find all weight-$2$ codewords in $\C$ and then to take $\C_0$ to be the punctured code obtained by restricting $\C$ to coordinates not in the support of any weight-$2$ codeword.
That allows us to reduce $\PCE$ on $\C, \C'$ to $\PCE$ on their reductions $\C_0, \C_0'$, which is analogous to the proof of \cref{thm:intro-low-R0-lattice-alg}. (However, unlike the analogous situation for lattices, this reduction to solving $\PCE$ on $\C, \C'$ is efficient.)

 We remark that the $2^{n_0 + o(n_0)}$ term in the running time of \cref{thm:intro-low-R0-code-alg} comes from calling the $\PCE$ algorithm of Babai~\cite{conf/soda/BabaiCGQ11} on the reductions $\C_0, \C_0'$ of $\C, \C'$.\footnote{One can improve this running time to $2^{n_0/2 + o(n_0)}$ using the randomized algorithm (resp., to $2^{n_0/3 + o(n_0)}$ using the quantum algorithm) from~\cite{journals/tit/BennettBBDLLW26}.}
 In fact, one can use~\cite{conf/soda/BabaiCGQ11} to achieve a similar running time for search $\PCE$ on direct sum codes $\C, \C' \cong \C_1 \oplus \cdots \oplus \C_m$ where each code $\C_i$ has block length $O(\log n)$. \cref{thm:intro-low-R0-code-alg} is an interesting special case of this because of its connection to structural properties of self-dual codes.

Again by analogy with \cref{item:intro-characteristic-vec} for self-dual lattices, every self-dual code $\C \in \F_2^n$ has characteristic vectors $\vec{w} \in \F_2^n$ (note that characteristic vectors of $\C$ are \emph{not} contained in $\C$).
Such a vector $\vec{w}$ is such that for every $\vec{v} \in \lat$, $\iprod{\vec{v}, \vec{w}} \equiv \half \wt(\vec{v}) \Mod{2}$, where $\wt(\vec{v})$ denotes the Hamming weight of $\vec{v}$. Similarly, one can define $s(\C)$ to be the minimum Hamming weight of a characteristic vector of $\C$, and one can show that $s(\C) = (n - 8K)/2$ for some $K \in \Z_{\geq 0}$.
Furthermore, one can define $\NC(K) := \sup \set{R_0(\C) : \C = \C^{\perp}, s(\C) \geq (n - 8K)/2}$. Again it follows by definition that if $s(\C) \geq (n - 8K)/2$ then $R_0(\C) \leq \NC(K)$.

As observed in~\cite{Elkies-long-shadows}, $\NC(K) \leq \NL(K)$ for every $K$ (see \cref{fact:Nc<Nl}).
So, by the line of work by Elkies and Gaulter on bounding $\NL(K)$, \cref{thm:intro-low-R0-code-alg} gives a polynomial-time algorithm for $\PCE$ on codes with $s(\C) \geq (n - 8K)/K$ for $K = 3$ unconditionally, and, assuming that $\NL(K)$ is bounded for every $K$, for any constant $K$. (Recall that~\cite{Elkies-long-shadows} asked whether this is true for $\NL(K)$. In fact, it also separately asked whether $\NC(K)$ is bounded for every $K$.)

We conclude by noting a consequence of \cref{thm:intro-low-R0-code-alg}.
The \emph{hull} of a code $\C$ is its intersection with its dual code (i.e., the hull of $\C$ is $\C \cap \C^{\perp}$), and it is well known that $\PCE$ is easy to solve on codes with small hulls both in theory and in practice~\cite{journals/tit/Sendrier00,conf/isit/BardetOS19}.
That is, hard instances of $\PCE$ necessarily consist of codes with large hulls. \cref{thm:intro-low-R0-code-alg} is interesting because it gives a natural class of codes $\C$ (ones for which $s(\C)$ is large) that have maximally large hulls but on which $\PCE$ is still easy.

\subsection{Open Questions}
A primary open question is whether one can find an algorithm for $\LIP$ on \emph{arbitrary} self-dual lattices that is faster than the $n^{O(n)}$-time algorithm for general $\LIP$ in~\cite{havivLatticeIsomorphismProblem2014}. Indeed, even a $2^{100n}$-time algorithm would be interesting.
One approach for this would be to compute shortest characteristic vectors for the input lattices $\lat$ and $\lat'$ (and therefore $s(\lat)$ and $s(\lat')$).
Indeed, one can think of $s(\lat)$ as a refinement of the genus of $\lat$ (the genus is certain efficiently computable data about $\lat$; see \cref{sec:prelims-genus}). And, $s(\lat)$ is computable in $2^{O(n)}$ time using known algorithms~\cite{journals/siamcomp/MicciancioV13,conf/focs/AggarwalDS15} for the Closest Vector Problem ($\CVP$), since it amounts to finding a shortest vector in a coset of $\lat$, which is equivalent to $\CVP$.
However, unfortunately, it is not clear how to use $s(\lat), s(\lat')$ to decide whether $\lat \cong \lat'$. Indeed, there are non-isomorphic odd self-dual lattices $\lat, \lat'$ with $s(\lat) = s(\lat')$ (see the pairs of such lattices of rank $n = 18$ and rank $n = 20$ in \cref{tab:n-8}), and it is easy to see that all even self-dual lattices have the all-zeroes vector $\vec{0}$ as their shortest characteristic vector (so that, e.g., shortest characteristic vectors cannot be used to distinguish between $E_8^{3k}$ and $\Lambda_{24}^k$).

Similarly, it is an interesting open question whether one can find faster algorithms for $\PCE$ on arbitrary self-dual codes than those in~\cite{conf/soda/BabaiCGQ11,journals/tit/BennettBBDLLW26} for arbitrary codes.

\subsection{Acknowledgments}

We thank Sasha Golovnev and Noah Stephens-Davidowitz for helpful discussions. In particular, we thank Noah Stephens-Davidowitz~\cite{PersComm-SDlat-StephensD26} for suggesting and allowing us to include his improvement to the $\coNP$ part of the protocol in \cref{thm:intro-coNP-proof}.
Our original protocol required the condition $\NL(K) = O(1)$ whereas the improved version from~\cite{PersComm-SDlat-StephensD26} relaxes this to $\NL(K) = O(\log n/\log \log n)$. The improved protocol is also more similar to our algorithm in \cref{thm:intro-low-R0-lattice-alg}.

%% file: prelims.tex
\section{Preliminaries}
\label{sec:prelims}

We use $\O_n := \set{O \in \R^{n \times n} : O^T O = I_n}$ to denote the set of $n \times n$ \emph{orthogonal matrices}. In particular, $\O_n$ consists of matrices $O \in \R^{n \times n}$ such that $O^T O = I_n$, hence left-multiplication by $O$ does not change the Euclidean norm of any vector in $\R^n$. $\GL_n(\Z)$ is the set of \emph{unimodular matrices}. That is, every matrix $U \in \GL_n(\Z)$ is an $n \times n$ integer matrix with $\det(U) = \pm 1$. We use $\SPD_n(R)$ to denote the set of symmetric positive definite matrices over a ring $R \subseteq \R$ (we will take $R = \Z$ or $R = \R$).

We use $\vec{e}_i \in \Z^n$ (or $\vec{e}_i \in \F_2^n$) for $i \in [n]$ to denote the $i$th standard normal basis vector.

\subsection{Lattices}
A (full-rank) \emph{lattice} $\lat \subset \R^n$ is the set of all integer linear combinations of some $n$ linearly independent vectors $\vec{b}_1, \ldots, \vec{b}_n \in \R^n$.
That is,
\[
\lat = \lat(\basis) := \set[\Big]{\sum_{i=1}^n a_i \vec{b}_i : a_1, \ldots, a_n \in \Z}
\]
is the lattice generated by the \emph{basis} $B := (\vec{b}_1, \ldots, \vec{b}_n) \in \GL_n(\R)$. 
We define the \emph{direct sum} of two lattices $\lat, \lat'$ to be the lattice obtained by concatenating vectors from $\lat$ and $\lat'$:
\[
\lat \oplus \lat' := \set{(\vec{v}^T, (\vec{v}')^T)^T : \vec{v} \in \lat, \vec{v}' \in \lat'} \ \text{.}
\]
The \emph{Gram matrix} of a basis $B$ is defined as $B^T B$,
and we call $Q := B^T B $ the \emph{quadratic form} associated with the lattice $\lat(B)$.

The \emph{minimum distance} of a lattice $\lat$ is defined as
\[
\lambda_1(\lat) := \min_{\vec{v} \in \lat \setminus \set{\vec{0}}} \norm{\vec{v}} \ \text{.}
\]
The \emph{determinant} of a lattice $\lat$ with basis $\basis$ is equal to $\det(\lat) := \det(B^T B)^{1/2}$.
Two bases $\basis, \basis' \in \GL_n(\R)$ generate the same lattice if and only if $\basis' = \basis U$ for some $U \in \GL_n(\Z)$, and so $\det(\lat)$ is well-defined. Minkowski's Theorem asserts that for any lattice $\lat \subset \R^n$,
\begin{equation} \label{eq:minkowski}
\lambda_1(\lat) \leq \sqrt{n} \cdot \det(\lat)^{1/n} \ \text{.}
\end{equation}
The \emph{dual lattice} $\lat^*$ of a lattice $\lat$ is defined as
\[
\lat^* := \set{ \vec{y} \in \lspan(\lat) : \forall \vec{x} \in \lat, \iprod{\vec{x}, \vec{y}} \in \Z} \ \text{.}
\]
It is straightforward to show that $\lat^*$ is indeed a lattice, $(\lat^*)^* = \lat$, and $\det(\lat^*) = 1/\det(\lat)$.

A lattice $\lat \subset \R^n$ is called \emph{integral} if $\lat \subseteq \lat^*$, i.e., if the inner product of any two vectors in $\lat$ is an integer. 
Lattices that satisfy $\lat = \lat^*$ are called \textit{self-dual}, and one can show that $\lat$ is self-dual if and only if it is integral and $\det(\lat) = 1$.
Self-dual lattices are often referred to as \textit{unimodular lattices} because a lattice $\lat(B)$ generated by a basis $B$ is self-dual if and only if its Gram matrix $B^T B$ is unimodular (see, e.g.,~\cite[Lemma 2]{Hunkenschroder-ZLIP-coNP}).

A self-dual lattice $\lat$ is called \emph{odd} if there exists $\vec{x} \in \lat$ such that $\norm{\vec{x}}^2$ is odd, and it is called \emph{even} otherwise \cite[p. 48]{Conway-Sloane-SPLAG}. Even self-dual lattices only exist in dimensions that are multiples of $8$ \cite[Ch. 7, Corollary 18]{Conway-Sloane-SPLAG}.

\subsection{Lattice Isomorphism and Quadratic Form Equivalence}
Given bases $\basis, \basis' \in  \GL_n(\R)$, the lattices $\lat(\basis)$ and $\lat(\basis')$ are \emph{isomorphic} if there exist an orthogonal matrix $O \in \O_n$ and a unimodular matrix $U \in \GL_n(\Z)$ such that $\basis' = O \basis U$. Equivalently, two lattices $\lat, \lat' \subset \R^n$ are isomorphic if there exists $O \in \O_n$ such that $O(\lat) = \lat'$. We define the equivalence class $[\lat] := \set{O(\lat) : O \in \O_n}$ of lattices isomorphic to $\lat$.

We also follow the notation of \cite{conf/eurocrypt/DucasW22} and denote the $R$-\emph{equivalence class} of a quadratic form $Q \in \SPD_n(\R)$ with respect to a ring $R$ as
\[
[Q]_R = \{U^T Q U: U \in \GL_n(R)\} \ \text{.}
\]
We are mainly interested in the case where $R = \Z$, and in this case we sometimes omit the subscript and simply write $[Q]$. 

\subsubsection{The Correspondence Between Lattices and Quadratic Forms} 
\label{sec:prelims-equiv-lip-qfe}

Two bases $\basis, \basis'$ generate isomorphic lattices if and only if their corresponding quadratic forms $Q := \basis^T \basis, ~Q' := (\basis')^T \basis' \in \SPD_n(\R)$ are $\Z$-equivalent. Indeed, if $O B U = B'$ for $O \in \O_n$ and $U \in \GL_n(\Z)$, then
\[
Q' = (B')^T B' = U^T B^T O^T O B U = U^T B^T B U = U^T Q U \ \text{.}
\]
Moreover, if $Q' = U^T Q U$ for $U \in \GL_n(\Z)$ then it must be the case that $B' = OBU$ for some $O \in \O_n$. Indeed, $O := B' (BU)^{-1}$ satisfies $OBU = B'$, and\footnote{We use the notation $A^{-T}$ for a matrix $A$ to mean its inverse transpose $(A^{-1})^T = (A^{T})^{-1}$.}
\[
O^T O
= (BU)^{-T} (B')^T B' (BU)^{-1} 
= B^{-T} U^{-T} Q' U^{-1} B^{-1}
= B^{-T} Q B^{-1}
= I_n \ \text{.}
\]

Quadratic forms $Q$ corresponding to integral lattices are useful to work with from a computational perspective since they contain all integer entries. So, we formalize $\LIP$ in terms of quadratic forms in the sequel.

We also note that the above discussion shows that the map $[\lat] \mapsto [Q]$ defined by taking $\basis$ to be a basis of $\lat$ and setting $Q := \basis^T \basis$ is a well-defined, bijective map from equivalence classes of lattices to equivalence classes of quadratic forms.
We say that a quadratic form $Q$ (respectively, quadratic form equivalence class $[Q]$) \emph{corresponds to} a lattice $\lat$ with basis $B$ (respectively, lattice equivalence class $[\lat]$) if $Q = B^T B$ (respectively, if $[Q] = [B^T B]$).

\subsubsection{Lattice Genus}
\label{sec:prelims-genus}

Taking $R = \Q$ and $R = \Z_p$, where $\Z_p$ denotes the $p$-adic integers, affords a coarser notion of equivalence.
Specifically, $[Q]_{\Z}= [Q']_{\Z}$ implies that $[Q]_{\Q} = [Q']_{\Q}$ and $[Q]_{\Z_p} = [Q']_{\Z_p}$ for all primes $p$, where $\Z_p$ denotes the $p$-adic integers.
The data $([Q]_\Q, ([Q]_{\Z_p})_p)$, where $p$ ranges over primes dividing $2 \det(Q)$, is called the \emph{genus} of $Q$. 
Moreover, the genus of $Q$ is efficient to compute given a prime factorization of $\det(Q)$; see, e.g., \cite{vanWoerden-genus-asisacrypt,conf/eurocrypt/DucasW22}.

So, quadratic forms $Q, Q'$ in distinct genera cannot be $\Z$-equivalent, and therefore $\LIP$ is essentially only a hard problem on quadratic forms in the same genus. 
There are at most two distinct genera that contain self-dual lattices in each dimension. In particular, the genera of full-rank, self-dual lattices of fixed dimension $n$ are the odd self-dual lattices, and, for $n \equiv 0 \Mod{8}$, the even self-dual lattices~\cite{Conway-Sloane-SPLAG}.

\subsection{Codes and Code Equivalence}

A \emph{binary, linear code} $\C$ (or simply \emph{code}) of length $n$ is a linear subspace of $\F_2^n$, and we call a code $\C \subseteq \F_2^n$ of dimension $k$ an $[n, k]$ code. We write $\clen(\C)$ and $\dim(\C)$ for the length and dimension of a code $\C$, respectively.
As with lattices, one can represent $[n, k]$ codes via a basis (i.e., \emph{generator matrix}) $G \in \F_2^{n \times k}$.\footnote{In this work, we use \emph{column} bases both for codes and lattices.}
Specifically, we define the code $\C(G)$ generated by $G$ to be
\[
    \C(G) := \{G \vec{x} : \vec{x} \in \F_2^k\} \ \text{.}
\]
Furthermore, we define the orthogonal subspace of a code $\C \subset \F_2^n$ to be the \emph{dual code} of $\C$, which we denote by $\C^\perp$. That is, 
\[
    \C^\perp := \set{\vec{y} \in \F_2^n : \forall \vec{c} \in \C, \iprod{\vec{c}, \vec{y}} = 0} 
    \ \text{.}
\]
By the rank-nullity theorem, if $\C$ is an $[n, k]$ code then $\C^{\perp}$ is an $[n, n - k]$ code.

We call a code $\C$ \emph{self-orthogonal} if $\C \subseteq \C^{\perp}$, which is equivalent to the condition that $\iprod{\vec{c}_1, \vec{c}_2} = 0$ for all codewords $\vec{c}_1, \vec{c}_2 \in \C$.
If $\C = \C^{\perp}$, then we call $\C$ \emph{self-dual}.
Note that a self-dual code $\C \subset \F_2^n$ must have dimension $k = n/2$, which in particular implies that self-dual codes only exist for even lengths $n$.

We define the \emph{direct sum} of two codes $\C, \C'$ to be the code
\[
\C \oplus \C' := \set{(\vec{c}^T, (\vec{c}')^T)^T : \vec{c} \in \C, \vec{c}' \in \C'} \ \text{.}
\]
We define the \emph{minimum distance} of a code $\C$ as
\[
d(\C) := \min_{\vec{c} \in \C \setminus \set{\vec{0}}} \wt(\vec{c}) \ \text{,}
\]
where $\wt(\cdot)$ denotes Hamming weight.

\subsubsection{Code Equivalence}
Let $\cP_n$ denote the set of $n \times n$ permutation matrices.
We say that two codes $\C, \C' \subseteq \F_2^n$ are \emph{permutationally equivalent}, denoted by $\C \permequiv \C'$, if there exists $P \in \cP_n$ such that $P (\C) = \set{P\vec{c} : \vec{c} \in \C} = \C'$. Recall that $\C(GU) = \C(G)$ for any generator matrix $G \in \F_2^{n \times k}$ and $U \in \GL_k(\F_2)$. So, we have that two codes $\C, \C'$ with respective generator matrices $G, G' \in  \F_2^{n \times k}$ are permutationally equivalent if and only if there exist a permutation matrix $P \in \cP_n$ and a full-rank matrix $U \in \GL_{k}(\F_2)$ such that $P G U = G'$ (equivalently, if there exists $P \in \cP_n$ such that $P(\C) = \C'$).

\subsubsection{Construction A}

Codes and lattices are closely related objects.
In particular, one can construct lattices from binary, linear codes via \emph{Construction A}.
\begin{definition}[Construction-A lattice] \label{def:consA}
    Given a linear code $\C \subset \F_2^n$, we define $\lat_\C := \C + 2 \Z^n$ as the \emph{Construction-A lattice} obtained from $\C$.
\end{definition}

Here we define arithmetic mixing elements of $\F_2$ and $\Z$ by implicitly lifting elements of $\F_2$ to their canonical representatives over the integers (i.e., to $\bit \subset \Z$).
It is straightforward to show that $\lat_\C \subset \R^n$ is a full-rank lattice.
We note that
    \begin{equation*}
        \lat_\C = \{\vec{y} \in \Z^n : \vec{y} \bmod 2 \in \C\} \ \text{,}
    \end{equation*}
where the $\bmod~2$ operation is applied component-wise to each vector.

The simplest self-dual code is the double repetition code
\begin{equation} \label{eq:double-rep-code}
\mathcal{Z} := \{(0,0)^T,(1,1)^T\} \in \F_2^2 \ \text{,}
\end{equation}
which is the code analogue of the lattice $\Z^2$.
(Here we are essentially following the notation of~\cite{Elkies-long-shadows} for $\cZ$, which is suggestive of its correspondence with $\Z$.)
The two objects are related via Construction A: $(1/\sqrt{2}) \cdot \lat_{\mathcal{Z}} \cong \Z^2$ and in fact $(1/\sqrt{2}) \cdot \lat_{\mathcal{Z}^r} \cong \Z^{2r}$ for any $r \in \Z_+$, where $\cZ^r := \cZ \oplus \cdots \oplus \cZ$.
More generally, applying Construction A to a self-dual code yields a scaled self-dual lattice; see \cref{lem:consA_s(L)}.

\subsection{Characteristic Vectors of Self-Dual Lattices and Codes}

We now introduce characteristic vectors of self-dual lattices and codes, which is the main technical tool that we leverage in this work. 
If $\lat$ is a self-dual lattice, a vector $\vec{w} \in \lat$ is called a \emph{characteristic vector} (also called a \emph{parity vector} in~\cite{Conway-Sloane-SPLAG}) if for all $\vec{v} \in \lat$, 
\[
\iprod{\vec{v}, \vec{w}} \equiv \iprod{\vec{v}, \vec{v}} \Mod 2 \ \text{.}
\]

Let $\wt(\vec{y})$ denote the Hamming weight of a vector $\vec{y} \in \F_2^n$.
 A characteristic vector of a self-dual binary code $\C \subset \F_2^n$ is a vector $\vec{w} \in \F_2^n$ that satisfies
\begin{equation*}
    \iprod{\vec{c}, \vec{w}} \equiv \frac{1}{2} \wt(\vec{c}) \Mod 2
\end{equation*}
for all $\vec{c} \in \C$.
We note that $\wt(\vec{c})$ is necessarily even for every codeword $\vec{c}$ in a self-dual code $\C$, and so $\frac{1}{2} \wt(\vec{c})$ is always an integer.
Unlike characteristic vectors of lattices, the characteristic vectors of a self-dual code are not themselves codewords.

We next present several important properties of characteristic vectors for self-dual lattices and codes.

\begin{fact}[Properties of characteristic vectors.] \label{fact:char-vec-properties}
Let $\lat \subset \R^n, \lat' \subset \R^{n'}$ be self-dual lattices, and let $\C \subset \F_2^n, \C' \subset \F_2^{n'}$ be self-dual codes.
\begin{enumerate}
    \item \label{item:char-vec-exists-shadow}
    The lattice $\lat$ (resp., code $\C$) has a characteristic vector $\vec{w}$ (resp., $\vec{w}')$. Moreover, the elements of the coset $2 \lat + \vec{w}$ for such $\vec{w}$ (resp., $\C + \vec{w}'$ for such $\vec{w}'$) are exactly the characteristic vectors of $\lat$ (resp, $\C$). 

    \item \label{item:char-vec-explicit}
    If $\basis = (\vec{b}_1, \ldots, \vec{b}_n)$ is a basis of $\lat$ with dual basis $\basis^* = (\vec{b}_1^*, \ldots, \vec{b}_n^*)$, then $\vec{w} := \sum_{i=1}^n \norm{\vec{b}_i^*}^2 \vec{b}_i$ is a characteristic vector of $\lat$.

    \item \label{item:char-vec-congruences} Every characteristic vector $\vec{w} \in \lat$ (resp., $\vec{w}'$ of $\C$) satisfies $\norm{\vec{w}}^2 \equiv n \Mod{8}$ (resp., $\wt(\vec{w}') \equiv n/2 \Mod{4}$).
    
    \item \label{item:char-vec-concatenation} The characteristic vectors of the lattice $\lat \oplus \lat'$ (resp., code $\C \oplus \C'$) are exactly the concatenations $(\vec{w}^T, (\vec{w}')^T)^T$ of characteristic vectors $\vec{w}$ for $\lat$ and $\vec{w}'$ for $\lat'$ (resp., $\vec{w}$ for $\C$ and $\vec{w}'$ for $\C'$).
    
\end{enumerate}
\end{fact}

\cref{item:char-vec-exists-shadow,item:char-vec-explicit,item:char-vec-congruences} are given in \cite[p. xxxiv]{Conway-Sloane-SPLAG}.
\cref{item:char-vec-explicit} is originally due to Gerstein~\cite{Gerstein04}. 
\cref{item:char-vec-concatenation} follows because elements of $\lat \oplus \lat'$ (resp., $\C \oplus \C'$) are concatenations of elements of $\lat, \lat'$ (resp., $\C, \C'$).
The coset $\lat + \half \vec{w}$ of $\lat$ (resp., $\C + \vec{w}'$ of $\C$) in \cref{item:char-vec-exists-shadow} is called the \emph{shadow} of $\lat$ (resp., $\C$).

Following the notation of Gaulter~\cite{Gaulter-S2}, we define $s(\lat)$ to be the squared Euclidean norm of a shortest characteristic vector of a self-dual lattice $\lat$.
\begin{equation} \label{eq:char-vec-squared-length}
s(\lat) := \min \set{\norm{\vec{w}}^2 : \vec{w} \in \lat \text{ and } \forall \vec{v} \in \lat, \iprod{\vec{v}, \vec{w}} \equiv \iprod{\vec{v}, \vec{v}} \Mod 2} \ \text{.}
\end{equation}
We also overload the notation $s(\cdot)$, and analogously define $s(\C)$ to be the minimum Hamming weight of a characteristic vector for a self-dual code $\C$.
\begin{equation} \label{eq:char-vec-codes}
s(\C) := \min \set{\wt(\vec{w}) : \vec{w} \in \F_2^n \text{ and } \forall \vec{c} \in \C, \iprod{\vec{c}, \vec{w}} \equiv \frac{1}{2} \wt(\vec{c}) \Mod 2} \ \text{.}
\end{equation}
It will be clear from context whether we are applying $s$ to a lattice versus a code.
We next observe several useful properties of $s$.
\begin{fact}[Properties of $s$.] \label{fct:properties-s}
  Let $\lat \subset \R^n, \lat' \subset \R^{n'}$ be self-dual lattices, and let $\C \subset \F_2^n, \C' \subset \F_2^{n'}$ be self-dual codes. The following hold:
    \begin{enumerate}
    \item \label{item:properties-s-orthogonal-permutation-matrix} If $O \in \O_n$ and $P \in \cP_n$, $s(O\lat) = s(\lat)$ and $s(P\C) = s(\C)$.
    \item \label{item:properties-s-congruence} $s(\lat) \equiv \rank(\lat) \Mod{8}$ and $s(\C) \equiv \dim(\C) \Mod{4}$.
    \item \label{item:properties-s-direct-sum} $s(\lat \oplus \lat') = s(\lat) + s(\lat')$ and $s(\C \oplus \C') = s(\C) + s(\C')$.
    \item \label{item:properties-s-Zr} For $n \in \Z_+$, $s(\Z^n) = n$ and $s(\cZ^n) = n$.
    \end{enumerate}
\end{fact}

\begin{proof}
\cref{item:properties-s-orthogonal-permutation-matrix} follows by noting that $\vec{w}$ (resp., $\vec{w}'$) is a characteristic vector of $\lat$ (resp., $\C$) if and only if $O \vec{w}$ (resp., $P \vec{w}'$) is a characteristic vector of $O \lat$ (resp., $P \C$).
\cref{item:properties-s-congruence} follows from \cref{fact:char-vec-properties}, \cref{item:char-vec-congruences}. 
\cref{item:properties-s-direct-sum} follows from \cref{fact:char-vec-properties}, \cref{item:char-vec-concatenation}. For \cref{item:properties-s-Zr}, first note that $\vec{1}^n$ is a characteristic vector of $\Z^n$ and $\vec{1}^n \otimes (1,0)^T$ is a characteristic vector of $\cZ^n$, where we use $\otimes$ to denote the Kronecker product. To see that there are no shorter characteristic vectors, observe that any vector $\vec{w} \in \Z^n$ (resp., $\vec{w}' \in \F_2^{2n}$) with $\norm{\vec{w}}^2 < n$ (resp., $\wt(\vec{w}') < n$) must have inner product $0$ with $\vec{e}_i \in \Z^n$ for some $i \in [n]$ (resp., $\vec{e}_i \otimes (1, 1)^T \in \F_2^{2n}$ for some $i \in [n]$).
\end{proof}

By \cref{fct:properties-s}, \cref{item:properties-s-direct-sum} and \cref{item:properties-s-Zr}, if $\lat = \lat_0 \oplus \Z^r$ for a reduced lattice $\lat_0$ (resp., $\C = \C_0 \oplus \cZ^r$ for a reduced code $\C_0$) then $s(\lat) = s(\lat_0) + r$ (resp., $s(\C) = s(\C_0) + r$).

We next note that given a code $\C \subset \F_2^n$, $\C$ is self-dual if and only if the scaled Construction A lattice $(1/\sqrt{2}) \cdot \lat_\C$ is self-dual. As discussed in \cite{Elkies-long-shadows}, there is a simple relationship between $s(\C)$ and $s(\lat)$ in this case.

\begin{lemma}
\label{lem:consA_s(L)}
 Given a code $\C \subseteq \F_2^n$, $\C$ is a self-dual code if and only if $\frac{1}{\sqrt{2}} \cdot \lat_\C \subset \R^n$ is a self-dual lattice.
 Furthermore, if $s(\C) = (n - 8K)/2$ for $K \in \Z_{\geq 0}$ then $s(\lat) = n - 8K$.
\end{lemma}

\begin{proof}
    The first part of the lemma follows from \cite[p. 183, Theorem 2]{Conway-Sloane-SPLAG}.
    For the ``furthermore'' part of the lemma, it is sufficient to show that $\lat' = (1/\sqrt{2}) \cdot (\C' + 2 \Z^n)$ is the shadow of $\lat$, where $\C'$ is the shadow of $\C$. In particular, if $(n - 8K)/2$ is the minimum Hamming weight of $C'$, the definition of Construction A guarantees that $(n - 8K)/4$ is the minimum squared Euclidean norm of a vector in $\lat'$.
    Indeed, assuming $\lat'$ is the shadow of $\lat$, $2 \lat'$ is the set of characteristic vectors of $\lat$. It follows that the minimum squared norm of a characteristic vector of $\lat$ is $n - 8K$.

    To prove the claim, consider $\vec{w} \in \lat'$ and $\vec{v} \in \lat$. By Construction A, $\vec{w} = (1/\sqrt{2}) \cdot (\vec{c}'+2\vec{z}_1)$ and $\vec{v} = (1/\sqrt{2}) \cdot (\vec{c}+2\vec{z}_2)$ for some $\vec{c}' \in \C'$, $\vec{c} \in \C$, and $\vec{z}_1, \vec{z}_2 \in \Z^n$.
    We then have that
    \[
    \iprod{\vec{v}, 2 \vec{w}}
 = \iprod*{\frac{1}{\sqrt{2}} \cdot (\vec{c} + 2 \vec{z}_2), \sqrt{2} \cdot (\vec{c}' + 2 \vec{z}_1)}
    \equiv \iprod{\vec{c}, \vec{c}'} \Mod{2} \ \text{,}
    \]
    and
    \[
    \iprod{\vec{v}, \vec{v}} 
   = \iprod*{\frac{1}{\sqrt{2}} \cdot (\vec{c} + 2 \vec{z}_2), \frac{1}{\sqrt{2}} \cdot (\vec{c} + 2 \vec{z}_2)}
    \equiv \half \cdot \iprod{\vec{c}, \vec{c}}
    \equiv \half \cdot \wt(\vec{c}) \Mod{2} \ \text{.}
    \]
   Combining these expressions, we get that $\iprod{\vec{c}, \vec{c}'} \equiv \half \cdot \wt(\vec{c}) \Mod{2}$ for all $\vec{c} \in \C$ if and only if $\iprod{\vec{v}, 2 \vec{w}} \equiv \iprod{\vec{v}, \vec{v}} \Mod{2}$ for all $\vec{v} \in \lat$. In other words, $2 \vec{w}$ is a characteristic vector of $\lat$ if and only if $\vec{c}'$ is a characteristic vector of $\C$.
\end{proof}

\subsection{Reduced Lattices and Codes}
Following~\cite[p.~$414$]{Conway-Sloane-SPLAG}, we define a \emph{reduced lattice} (also called an \emph{initial lattice}) $\lat_0$ to be a self-dual lattice with $\lambda_1(\lat_0)^2 \geq 2$.
Similarly, we define a \emph{reduced code} $\C_0$ to be a self-dual code with $d(\C_0) \geq 4$.

Importantly, every self-dual lattice (resp., code) is isomorphic (resp., permutationally equivalent) to the direct sum of a reduced lattice (resp., reduced code) and $\Z^r$ (resp., $\cZ^r$) for some $r \geq 0$.

\begin{fact}[{\cite[p. 414]{Conway-Sloane-SPLAG}, \cite{Elkies-long-shadows}}]
\label{fct:reduced-direct-sum}
The following hold.
\begin{enumerate}
\item Every self-dual lattice $\lat \subset \R^n$ is isomorphic to $\lat_0 \oplus \Z^r$ for a reduced lattice $\lat_0$ (which is unique up to isomorphism) and some $r \geq 0$.
\item Every self-dual code $\C \subset \F_2^n$ is isomorphic to $\C_0 \oplus \cZ^r$ for a reduced code $\lat_0$ (which is unique up to permutational equivalence) and some $r \geq 0$.
\end{enumerate}
\end{fact}

We call a lattice $\lat_0$ (resp., $\C_0$) such that $\lat \cong \lat_0 \oplus \Z^r$ (resp., $\C \cong \C_0 \oplus \cZ^r$) as in \cref{fct:reduced-direct-sum} a \emph{reduction} of $\lat$ (resp., $\C$).
We define $R_0(\lat)$ (resp., $R_0(\C)$) to be the rank of a reduction of $\lat$ (resp., length of a reduction of $\C$). Note that $R_0$ is well-defined.

We will use the following definition, which appears in \cite{Elkies-long-shadows}.
\begin{equation} \label{eq:NLK}
\NL(K) := \sup \; \{R_0(\lat) : \lat = \lat^*, s(\lat) \geq \rank(\lat) - 8K \} \ \text{.}
\end{equation}
Equivalently, $\NL(K)$ is the maximum rank of a full-rank reduced lattice $\lat_0 \subset \R^n$ with $s(\lat_0) \geq n - 8K$, or infinity if there are such lattices $\lat_0$ of arbitrarily high rank $n$.
We also define the analogous quantity to $\NL(K)$ for codes.\footnote{The quantities defined in \cref{eq:NLK,eq:NCK} appear in~\cite{Elkies-long-shadows} as $N_K$ and $n_K$, respectively.}
\begin{equation} \label{eq:NCK}
\NC(K) := \sup \; \{R_0(\C) : \C_0 = \C_0^\perp, s(\C) \geq (\clen(\C) - 8K)/2 \} \ \text{.}
\end{equation}
Equivalently, $\NC(K)$ is the maximum length of a reduced code $\C_0 \subset \F_2^n$ with $s(\C_0) \geq (n - 8K)/2$, or infinity if there are such codes $\C_0$ with arbitrarily high length $n$.

The following fact, noted by Elkies in \cite{Elkies-long-shadows}, is a corollary of \cref{lem:consA_s(L)}.
\begin{fact}
\label{fact:Nc<Nl}
For all $K \in \Z_{\geq 0}$, $\NC(K) \leq \NL(K)$.
\end{fact}

\input{nminus8-table}

In a beautiful line of work, Elkies and Gaulter proved the following bounds on $\NL(K)$. See also \cref{tab:n-8}.

\begin{theorem}
\label{thm:Nkbounds}
    The following bounds on $\NL(K)$ are known.
    \begin{enumerate}
        \item $\NL(0) = 0$ by~\cite{Elkies-char-Zn} (i.e., there are no reduced lattices $\lat_0$ of rank $n$ with $s(\lat_0) = n$).
        \item $\NL(1) = 23$ by~\cite{Elkies-long-shadows}.
        \item \label{item:N2-bound} $\NL(2) \leq 89$ by~\cite[Theorem~1.1]{Gaulter-S2}.
        \item $\NL(3) \leq 8388630$ by~\cite[Theorem~4.5]{Gaulter-Lattices-Without-Short}.
    \end{enumerate}
\end{theorem}

We note in passing that the earlier work of Gaulter~\cite{Gaulter-Lattices-Without-Short} obtained the weaker bound $\NL(2) \leq 2907$, which nevertheless would suffice for our purposes. We similarly note in passing that~\cite{gaulter1998characteristic,Gaborit2007} give improved bounds on the analog of $\NL(K)$ for self-dual lattices $\lat$ with $\lambda(\lat)^2 = K + 1$, which they call extremal.
Furthermore, we note that the bounds in \cref{thm:Nkbounds} for $\NL(K)$ with $K = 0, 1, 2, 3$ also hold for $\NC(K)$ by \cref{fact:Nc<Nl}. 

In~\cite{Elkies-long-shadows}, Elkies asked whether each of the quantities $\NL(K)$ and $\NC(K)$ is finite for all $K$. We state the positive resolution of these questions as the following two hypotheses.
\begin{hypothesis}
\label{hyp:finite_NLK}
For all $K \in \Z_{\geq 0}$, $\NL(K) < \infty$.
\end{hypothesis}

\begin{hypothesis}
\label{hyp:finite_NCK}
For all $K \in \Z_{\geq 0}$, $\NC(K) < \infty$.
\end{hypothesis}

Again using \cref{fact:Nc<Nl}, we have that \cref{hyp:finite_NLK} implies \cref{hyp:finite_NCK} (as noted in~\cite{Elkies-long-shadows}).

\subsection{Computational Problems}
\label{sec:prelims-comp-probs}

We first define a certain decision version of $\LIP$ for self-dual lattices with long shortest characteristic vectors. We discuss the motivation for our particular definition in \cref{sec:intro-our-results}.

\begin{definition}[Decision $\SDLIP_K$]
\label{def:SDLIP}
For $K = K(n) \geq 0$, we define $\SDLIP_K$ as follows. 
The input is a pair of quadratic forms $Q, Q' \in \SPD_n(\Z) \cap \GL_n(\Z)$ corresponding to self-dual lattices $\lat, \lat' \subset \R^n$, respectively, and the goal is to decide between the following two cases.
\begin{itemize}
    \item (YES case.) $[Q]_{\Z} = [Q']_{\Z}$ and $s(\lat) \geq n - 8K$. 
    
    \item (NO case.) $[Q]_{\Z} \neq [Q']_{\Z}$, and either $s(\lat) \geq n - 8K$ or $s(\lat') \geq n - 8K$ (or both).
\end{itemize}
\end{definition}

We will also study versions of search $\LIP$ and search Permutational Code Equivalence ($\PCE$) on self-dual lattices and additional promises on their inputs. Here, we simply recall the definitions of search $\LIP$ and search $\PCE$.

\begin{definition}[Search $\LIP$]
\label{def:SearchLIP}
We define \emph{search $\LIP$} as follows.
The input is a pair of quadratic forms $Q, Q' \in \SPD_n(\Z)$ with $[Q]_{\Z} = [Q']_{\Z}$, and the goal is to output a unimodular matrix $U \in \GL_n(\Z)$ such that $U^T Q U = Q'$.
\end{definition}

\begin{definition}[Search $\PCE$]
\label{def:SearchPCE}
We define \emph{search $\PCE$} as follows.
The input is a pair of generator matrices $G, G' \in \F_2^{n \times k}$ of respective codes $\C, \C'$ with $\C \permequiv \C'$ for some integers $k \leq n$. The goal is to output a permutation matrix $P \in \cP_n$ such that $P(\C) = \C'$.
\end{definition}

\subsection{Lattice and Code Algorithms}
\label{sec:prelims-lattice-alg}

We next give several algorithms for lattice and coding problems.
Recall that the Shortest Vector Problem ($\SVP$) on lattices is, given a basis of a lattice $\lat$ as input, to find a shortest non-zero vector $\vec{v} \in \lat$ (i.e., a non-zero vector $\vec{v} \in \lat$ with $\norm{\vec{v}} = \lambda_1(\lat)$). The Closest Vector Problem ($\CVP$) is, given a basis of a lattice $\lat$ and a target vector $\vec{t}$ as input, to find a closest lattice vector $\vec{v}$ in $\lat$ to $\vec{t}$ (i.e., a vector $\vec{v} \in \lat$ with $\norm{\vec{v} - \vec{t}} = \min_{\vec{v}' \in \lat} \norm{\vec{v}' - \vec{t}}$).

\begin{theorem}[{\cite{journals/siamcomp/MicciancioV13}}] \label{thm:MV-det-svp-cvp}
There are $4^{n + o(n)}$-time algorithms for $\SVP$ and $\CVP$ on lattices of rank $n$.
\end{theorem}

 We then get the following corollary about computing shortest characteristic vectors of lattices.

\begin{corollary} \label{cor:comp-shortest-char-vec}
There is a $4^{n + o(n)}$-time algorithm for computing a shortest characteristic vector of a self-dual lattice $\lat$ of rank $n$. In particular, there is a $4^{n + o(n)}$-time algorithm to compute $s(\lat)$.
\end{corollary}

\begin{proof}
By \cref{fact:char-vec-properties}, \cref{item:char-vec-explicit}, there is an efficient construction of a (not necessarily shortest) characteristic vector $\vec{w}$ of $\lat$. By \cref{fact:char-vec-properties}, \cref{item:char-vec-exists-shadow} the elements of the lattice coset $2\lat + \vec{w}$ are exactly the characteristic vectors of $\lat$.
Finding a shortest vector in a lattice coset is equivalent to $\CVP$, and so the corollary follows from \cref{thm:MV-det-svp-cvp}.
\end{proof}

We next give the algorithm of Haviv and Regev~\cite{havivLatticeIsomorphismProblem2014} for (general) $\LIP$.

\begin{theorem}[{\cite{havivLatticeIsomorphismProblem2014}}] \label{thm:HR-LIP}
There is a $n^{O(n)}$-time algorithm for $\CVP$ on lattices of rank $n$.
\end{theorem}

Finally, we give the algorithm of Babai~\cite{conf/soda/BabaiCGQ11} for (general) $\PCE$.\footnote{Although~\cite{conf/soda/BabaiCGQ11} has four authors, the result that we cite is attributed solely to Babai.}

\begin{theorem}[{\cite{conf/soda/BabaiCGQ11}}] \label{thm:Babai-PCE}
There is a $2^{n + o(n)}$-time algorithm for $\PCE$ on codes of length $n$.
\end{theorem}

\subsection{Computational Lemmas}

We next present two computational lemmas, the first of which is from~\cite{Hunkenschroder-ZLIP-coNP}. We use $\vec{e}_j \in \R^n$ for $j \in [n]$ to denote the $j$th standard normal basis vector.
\begin{lemma}[{\cite[Lemma 6]{Hunkenschroder-ZLIP-coNP}}]
\label{lem:checkz}
Let $Q \in \SPD_n(\Z) \cap\GL_n(\Z)$.
Then
\begin{equation*}
    \vec{e}_j^T Q \vec{e}_j \equiv \vec{e}_j^T Q \vec{z} \Mod{2}
\end{equation*}
holds for all $j = 1, \ldots, n$
if and only if for every basis $B$ with $B^T B = Q$, the vector $\vec{w} = B \vec{z}$ is a characteristic vector of the lattice $\lat(B)$.
\end{lemma}

We next prove norm bounds that we will use to bound the bit length of the certificates that we will use in our protocol for $\SDLIP$.
\cref{item:norm-bound-vec} is similar to \cite[Lemma 7]{Hunkenschroder-ZLIP-coNP}.
In the following, we use $\norm{A}$ to denote the operator norm of a matrix $A$ and $\norm{A}_{\infty}$ to denote the maximum magnitude of an entry in $A$.
\begin{lemma} \label{lem:norm-bounds}
Let $Q, Q' \in \SPD_n(\Z)$. Let $\len(\cdot)$ denote the bit length of a vector or matrix.
\begin{enumerate}
\item \label{item:norm-bound-vec} Let $\beta > 0$, let $\vec{z} \in \R^n$, and assume that $\vec{z}^T Q \vec{z} \leq \beta$. Then $\norm{\vec{z}} \leq (\beta \norm{Q^{-1}})^{1/2}$. In particular, if $\vec{z} \in \Z^n$ then $\len(\vec{z}) \leq \poly(n, \len(Q), \log \beta)$.\footnote{Recall that $\norm{\cdot}$ denotes the 2-norm.}
\item \label{item:norm-bound-mat} Suppose that $U \in \GL_n(\R)$ is such that $U^T Q U = Q'$. Then $\norm{U} \leq (\norm{Q'} \norm{Q^{-1}})^{1/2}$. In particular, if $U \in \GL_n(\Z)$ then $\len(U) \leq \poly(n, \len(Q), \len(Q'))$.
\end{enumerate}
\end{lemma}

\begin{proof}
Let $A, B \in \GL_n(\R)$ be such that $Q = A^T A$ and $Q' = B^T B$.
For \cref{item:norm-bound-vec}, we have that
\[
\norm{\vec{z}} \leq \norm{B^{-1}} \norm{B \vec{z}} = \norm{Q^{-1}}^{1/2} \norm{B \vec{z}} \leq (\beta \norm{Q^{-1}})^{1/2} \ \text{.}
\]
For \cref{item:norm-bound-mat}, we have that $U^T A^T A U = U^T Q U = B^T B$, and therefore that $O := B U^{-1} A^{-1}$ is an orthogonal matrix. It follows that
\[
\norm{U} = \norm{A^{-1} O^{-1} B} \leq \norm{A^{-1}} \norm{O^{-1}} \norm{B} = \norm{A^{-1}} \norm{B} = (\norm{Q^{-1}} \norm{Q'})^{1/2} \ \text{.}
\] 

Let $P = (\vec{p}_1, \ldots, \vec{p}_n) \in \SPD_n(\Z)$. Then $\norm{P} \leq n \cdot \norm{P}_{\infty}$.
Furthermore, $P^{-1} = \adj(P)/\det(P)$ by Cramer's rule, where $\adj(P)$ is the adjugate matrix of $P$, and $\norm{\adj(P)}_{\infty} \leq \max_{i \in [n]} \norm{\vec{p}_i}^{n - 1} \leq \norm{P}^n$, where the first inequality uses Hadamard's inequality. So,
\[
\norm{P^{-1}} \leq \norm{P}^n/\abs{\det(P)} \leq \norm{P}^n \ \text{,}
\]
where the second inequality uses the fact that $\det(P)$ is integer-valued and non-zero.
Bounds on the magnitudes of the (integer) entries and therefore bit lengths of $\vec{z}$ and $U$ in the ``in particular'' statements follow.
\end{proof}

%% file: nminus8-table.tex
\begin{table}[t]
\centering
\renewcommand{\arraystretch}{1.35}
\small
\begin{tabular}{c|c|c|c|c|c|c|c|c|c|c|c|c|c|c}
$n$ & 8 & 12 & 14 & 15 & 16 & 17 & 18 & 18 & 19 & 20 & 20 & 21 & 22 & 23 \\
\hline
\textbf{Lattice} $\lat$ & $E_8$ & $D_{12}$ & $E_7^2$ & $A_{15}$ & $D_8^2$ & $A_{11}E_6$ & $D_6^3$ & $A_9^2$ & $A_7^2D_5$ & $D_4^5$ & $A_5^4$ & $A_3^7$ & $A_1^{22}$ & $O_{23}$ \\
\hline
\textbf{Code} $\C$ & $A_8$ & $B_{12}$ & $D_{14}$ & --- & $F_{16}$ & --- & $H_{18}$ & --- & --- & $M_{20}$ & --- & --- & $G_{22}$ & --- \\
\end{tabular}
\caption{\small A complete list (up to isomorphism) of the $14$ full-rank reduced lattices $\lat \subset \R^n$ with $s(\lat) = n - 8$ given in~\cite{Elkies-long-shadows}. When applicable, we also include the corresponding reduced binary code $\C \subset \F_2^n$ for which $\lat \cong (1/\sqrt{2}) \lat_{\C}$ (i.e., $\lat$ can be constructed from $\C$ via Construction A).
We follow the notation of~\cite{Elkies-long-shadows} and \cite[Table 16.7]{Conway-Sloane-SPLAG} by specifying self-dual lattices $\lat$ in terms of their constituent root lattices (which are adjoined by so-called glue vectors to make them self-dual).
The listed codes $\C$, not all of which appear in~\cite{Elkies-long-shadows}, are (up to permutational equivalence) the 
unique reduced codes of a given length $n$ with the minimum possible number of weight-$4$ vectors (according to \cite[Theorem 1A]{Elkies-long-shadows}).
Bases for the codes $A_8, B_{12}, D_{14}, F_{16}, H_{18}$, and $M_{20}$ appear in~\cite[Table 1]{pless-selforthcodes-1972}.
A basis for $G_{22}$ is given in \cite{An2025ShortestSOEmbeddings}.}
\label{tab:n-8}
\end{table}

%% file: SDLIP.tex
\section{On (Non-)Isomorphism of Self-Dual Lattices}
\label{sec:SDLIP}

In this section, we study $\LIP$ on self-dual lattices.

\subsection{An Algorithm for \texorpdfstring{$\LIP$}{LIP} on Self-Dual Lattices}

We first argue that there is a roughly $2^{n/2}$-time algorithm for computing the decomposition of a self-dual lattice $\lat$ into the direct sum of $\lat_0$ and a rotation of $\Z^r$.
The algorithm works by noting that the problem reduces to finding unit-length vectors in $\lat$, and so we start with the following theorem. Recall that an \emph{integral lattice} $\lat$ is one that satisfies $\lat \subset \lat^*$ (and therefore that every self-dual lattice is integral).

\begin{lemma} \label{lem:lip-integral-lat-alg}
There is a $2^{n/2 + o(n)}$-time randomized algorithm for $\SVP$ on integral lattices $\lat \subset \R^n$ with $\lambda_1(\lat) \leq o(\sqrt{n}/\log n)$.
\end{lemma}

\begin{proof}
A lattice $\lat \subset \R^n$ is called \emph{semi-stable} if for every sublattice $\lat' \subseteq \lat$, $\det(\lat') \geq 1$. And, \cite[Corollary 5.5]{conf/eurocrypt/BennettGPS23} gives a $2^{n/2 + o(n)}$-time randomized algorithm for $\SVP$ on semi-stable lattices $\lat \subset \R^n$ with $\lambda_1(\lat) = o(\sqrt{n}/\log n)$. So, to prove the lemma it suffices to show that an arbitrary integral lattice $\lat$ is semi-stable.

For this, consider any arbitrary sublattice $\lat' \subseteq \lat$ of $\lat$, and let $\basis'$ be a basis of $\lat'$. We have that $\det(\lat')^2 = \det((\basis')^T \basis')$ by definition, and, because $\lat'$ is integral, $(\basis')^T \basis'$ is integer-valued. Furthermore, the Gram matrix of a basis is positive definite and so $\det((\basis')^T \basis') > 0$. It follows that $\det(\lat')^2 = \det((\basis')^T \basis')$ is an integer greater than $0$, which implies that $\det(\lat') \geq 1$, as needed.
\end{proof}

We then get the following theorem.

\begin{theorem} \label{thm:computing-reduced-lattice}
There is a randomized algorithm that, on input a quadratic form $Q \in \SPD_n(\Z)$ corresponding to a self-dual lattice $\lat \subset \R^n$ with $\lat \cong \lat_0 \oplus \Z^r$ for a reduced lattice $\lat_0$, outputs a unimodular matrix $U \in \GL_n(\Z)$ such that $U^T Q U = Q_0 \oplus I_r$, where $Q_0$ is a quadratic form corresponding to $\lat_0$. The algorithm runs in $2^{n/2 + o(n)}$ time.
\end{theorem}

\begin{proof} For ease of presentation, we give the algorithm in terms of lattices rather than quadratic forms. Let $\basis$ be a basis of $\lat$.

The algorithm works as follows. It first initializes a set of vectors $V := \emptyset$. It then calls the algorithm from \cref{lem:lip-integral-lat-alg} on $\basis$, receiving as output $\vec{v} \in \lat$. If $\norm{\vec{v}} = 1$, the algorithm sets $V := V \cup \set{\vec{v}}$ and then recurses on (a basis of) $\pi_{\lspan(V)^{\perp}}(\lat)$. Otherwise, it does nothing. (We assume without loss of generality that the $\SVP$ algorithm in \cref{lem:lip-integral-lat-alg} always outputs \emph{some} lattice vector, even if the promise that $\lambda_1(\lat) = o(\sqrt{n}/\log n)$ is not met.)
Let $W$ be a basis of $\pi_{\lspan(V)^{\perp}}(\lat)$ for the final set $V$ of unit vectors. The algorithm then computes and outputs $U \in \GL_n(\Z)$ such that $BU = (W | V)$.

The running time is clear since each of the at most $n$ calls made to \cref{lem:lip-integral-lat-alg} takes $2^{n/2 + o(n)}$ time, and all of the linear algebraic operations performed are efficient.
We now prove correctness.
By definition of a reduced lattice, $\lambda_1(\lat_0) > 1$. 
When $r = 0$ there are no unit vectors in $\lat$, and when $r \geq 1$, the shortest non-zero vectors in $\lat \cong \lat_0 \oplus \Z^r$ are exactly the $2r$ unit vectors $\pm \vec{v}_1, \ldots, \pm \vec{v}_r$, with $\iprod{\vec{v}_i, \vec{v}_j} = 0$ for $i \neq j$ such that $\lat(\vec{v}_1, \ldots, \vec{v}_r) \cong \Z^r$.
So, when $r \geq 1$, $\lambda_1(\lat) = 1$ and the premise of the $\SVP$ algorithm in \cref{lem:lip-integral-lat-alg} is met. It therefore necessarily outputs $\pm \vec{v}_i$ for some $i$.
It follows that at recursion depth $r'$ for $0 \leq r' \leq r$, $\lat(V) \cong \Z^{r'}$ and $\pi_{\lspan(V)^{\perp}}(\lat) \cong \lat_0 \oplus \Z^{r - r'}$. In particular, at recursion depth $r$, $\lat(V) \cong \Z^r$ and $\pi_{\lspan(V)^{\perp}}(\lat) \cong \lat_0$, as needed.

\end{proof}

We now restate and prove our main algorithm for search $\SDLIP$ on self-dual lattices $\lat, \lat'$, which is relatively efficient if $R_0(\lat), R_0(\lat')$ are small.

\lowRlatalg*

\begin{proof}
The algorithm works as follows. First, it uses the algorithm from \cref{thm:computing-reduced-lattice} to compute unimodular matrices $U, U' \in \GL_n(\Z)$ such that $U^T Q U = Q_0 \oplus I_r$ and $(U')^T Q' U' = Q_0' \oplus I_r$ for quadratic forms $Q_0, Q_0' \in \SPD_{n_0}(\Z)$ corresponding to reductions of $\lat, \lat'$, respectively, and $r := n - n_0$. It then uses the Haviv-Regev $\LIP$ algorithm given in \cref{thm:HR-LIP} to solve search $\LIP$ on $Q_0, Q_0'$, obtaining $U_0 \in \GL_n(\Z)$ such that $U_0^T Q_0 U_0 = Q_0'$. 
Finally, it outputs $U_f := U (U_0 \oplus I_r) (U')^{-1} \in \GL_n(\Z)$.

By the assumption that $\lat \cong \lat'$ and the guarantees of \cref{thm:computing-reduced-lattice,thm:HR-LIP}, $U_f^T Q U_f = Q'$, showing correctness. Furthermore, \cref{thm:computing-reduced-lattice} runs in $2^{n/2 + o(n)}$ time on each of $Q$ and $Q'$, and the search $\LIP$ algorithm in \cref{thm:HR-LIP} runs in $n_0^{O(n_0)}$ time on $Q_0, Q_0'$. The running time bound follows.
\end{proof}

\subsection{Certifying Non-Isomorphism of Self-Dual Lattices}

In this section, we show that $\SDLIP_K$ is in $\coNP$ for appropriate $K$. We also show that $\SDLIP_K$ is in $\NP$ for appropriate $K$, which is non-trivial because YES instances of $\SDLIP_K$ on lattices $\lat, \lat'$ have the property that $s(\lat) = s(\lat')$ must be large.

We first show how to certify that a lattice $\lat_0$ is the reduction of a self-dual lattice $\lat$.
\begin{lemma} \label{lem:certifying-reduction}
Let $Q \in \SPD_n(\Z)$, $Q_0 \in \SPD_{n_0}(\Z)$ with $n_0 \leq n$ be quadratic forms corresponding to self-dual lattices $\lat$, $\lat_0$, respectively.
There is an algorithm for deciding whether $\lat_0$ is a reduction of $\lat$ that runs in $4^{n_0 + o(n_0)} + \poly(n)$ time.
In particular, if $n_0 = O(\log n)$ then the algorithm runs in $\poly(n)$ time.
\end{lemma}

\begin{proof}
The certificate is a unimodular matrix $U \in \GL_n(\Z)$. The verifier first checks that $U^T Q U = Q_0 \oplus I_r$ for some $r \in \Z_{\geq 0}$. It then checks that $Q_0$ does not represent $1$ (i.e., that $\lambda_1(\lat_0) > 1$) using the deterministic $\SVP$ algorithm in \cref{thm:MV-det-svp-cvp}. The verifier accepts if and only if both of these checks pass.

The $\SVP$ algorithm in \cref{thm:MV-det-svp-cvp} runs in $4^{n_0 + o(n_0)}$ time on $Q_0$, and so it is clear that the verifier runs in the claimed time bound. It remains to argue correctness. The first check passes if and only if $\lat \cong \lat_0 \oplus \Z^r$ for some lattice $\lat_0$, and the second check passes if and only if $\lat_0$ is a reduced lattice. 
It follows that $\lat_0$ is a reduction of $\lat$ if and only if both checks pass, as needed.
\end{proof}

We get the following lemmas showing that $\SDLIP_K$ is in $\NP$ and $\coNP$, respectively, for sufficiently small $K = K(n)$.

\begin{lemma} \label{lem:SDLIP_K-NP}
For any $K = K(n)$ such that $\NL(K) \leq O(\log n)$, $\SDLIP_K$ is in $\NP$.
\end{lemma}

\begin{proof}
Let $Q, Q' \in \SPD_n(\Z)$ be the input quadratic forms corresponding to self-dual lattices $\lat, \lat'$. The certificate consists of unimodular matrices $U, U' \in \GL_n(\Z)$, and $U_0 \in \GL_{n_0}(\Z)$, where $n_0 := R_0(\lat) = R_0(\lat')$.
The verifier first computes $U^T Q U = Q_0 \oplus I_r$ and $(U')^T Q' U' = Q_0' \oplus I_r$, where $r \in \Z_{\geq 0}$.
It then checks that $Q_0$ and $Q_0'$ are quadratic forms corresponding to reductions $\lat_0$ and $\lat_0'$ of $\lat$ and $\lat'$, respectively, using the algorithm in \cref{lem:certifying-reduction}. Finally, it checks that $U_0^T Q_0 U_0 = Q_0'$.

From the checks using \cref{lem:certifying-reduction}, we know that $\lat_0$ and $\lat_0'$ are reductions of $\lat$ and $\lat'$, respectively. Furthermore, the fact that $U_0^T Q_0 U_0 = Q_0'$ implies that $\lat_0 \cong \lat_0'$, and therefore that $\lat \cong \lat'$. Correctness follows.
Moreover, because $n_0 \leq \NL(K) = O(\log n)$ by assumption, the algorithm runs in $4^{n_0 + o(n_0)} + \poly(n) \leq \poly(n)$ time by \cref{lem:certifying-reduction}, as needed.
\end{proof}

\begin{lemma} \label{lem:SDLIP_K-coNP}
For any $K = K(n)$ such that $\NL(K) \leq O(\log n/\log \log n)$, $\SDLIP_K$ is in $\coNP$.
\end{lemma}

\begin{proof}
Let $C > 0$ be a constant such that $\NL(K) < C \log n/\log \log n$ for all $n$.
Let $Q, Q' \in \SPD_n(\Z)$ be the input quadratic forms corresponding to self-dual lattices $\lat, \lat'$.
If $Q, Q'$ is a NO instance of $\SDLIP_K$, then by definition we have that $\lat \not\cong \lat'$ and that one of the the following two cases holds.

\textbf{Case 1:} $s(\lat) \geq n - 8K$ and $s(\lat') \geq n - 8K$. In this case, the certificate consists of unimodular matrices $U, U' \in \GL_n(\Z)$. The verifier first computes $U^T Q U = Q_0 \oplus I_r$ and $(U')^T Q' U' = Q_0' \oplus I_{r'}$, where $r, r' \in \Z_{\geq 0}$.
It rejects if $n - r > C \log n/\log \log n$ or $n - r' > C \log n/\log \log n$. It then checks that $Q_0$ and $Q_0'$ are quadratic forms corresponding to reductions $\lat_0$ and $\lat_0'$ of $\lat$ and $\lat'$, respectively, using the algorithm in \cref{lem:certifying-reduction}.
If $r \neq r'$ the verifier accepts. Otherwise, the verifier runs the $\LIP$ algorithm in \cref{thm:HR-LIP} on $Q_0$ and $Q_0'$, and it accepts if and only if the $\LIP$ algorithm asserts that $\lat_0 \not\cong \lat_0'$.

\textbf{Case 2:} $s(\lat) \geq n - 8K$ and $s(\lat') < n - 8K$, or $s(\lat) < n - 8K$ and $s(\lat') \geq n - 8K$.
Assume without loss of generality that $s(\lat) \geq n - 8K$ and $s(\lat') < n - 8K$.
In this case, the certificate consists of $U \in \GL_n(\Z)$ and $\vec{z} \in \Z^n$.
The verifier first computes $U^T Q U = Q_0 \oplus I_r$, where $r \in \Z_{\geq 0}$. It rejects if $n - r > C \log n/\log \log n$.
Second, it checks that $Q_0$ is a quadratic form corresponding to a reduction $\lat_0$ of $\lat$ using the algorithm in \cref{lem:certifying-reduction}.
Third, it computes $s(\lat_0)$ using the algorithm in \cref{cor:comp-shortest-char-vec}. Fourth, the verifier checks that $\vec{z}$ is the coefficient vector of a characteristic vector of $\lat'$ using \cref{lem:checkz}. Finally, it accepts if and only if $\vec{z}^T Q' \vec{z} < s(\lat_0) + r$.

Correctness follows in Case 1 because $\lat \not \cong \lat'$ if either $\rank(\lat_0) \neq \rank(\lat_0')$ or $\rank(\lat_0) = \rank(\lat_0')$ but $\lat_0 \not \cong \lat_0'$.
Correctness follows in Case 2 because
\[
s(\lat') \leq \vec{z}^T Q' \vec{z} < s(\lat_0) + r = s(\lat) \ \text{,}
\]
where the equality uses \cref{fct:properties-s}, \cref{item:properties-s-direct-sum}.
Therefore, $\lat \not \cong \lat'$ since $s$ is preserved under isomorphism (\cref{fct:properties-s}, \cref{item:properties-s-orthogonal-permutation-matrix}).

We next analyze the verifier's running time. 
First, we note that the certificate in either case has polynomial bit length by \cref{lem:norm-bounds}.
In each case, checking that $Q_0$ corresponds to a reduced lattice $\lat_0$ using \cref{lem:certifying-reduction} takes at most $2^{\NL(K)} \leq 2^{C \log n/\log \log n} = \poly(n)$ time. Similarly, checking that $Q_0'$ is reduced in Case 1 takes $\poly(n)$ time.
Let $n_0 := n - r$, and note that $n_0 \leq \NL(K) \leq C \log n/\log \log n$ by assumption.
In Case 1, checking that $\lat \not\cong \lat'$ using the $\LIP$ algorithm in \cref{thm:HR-LIP} takes
\[
n_0^{O(n_0)} 
\leq (C \log n/\log \log n)^{O(C \log n/\log \log n)} \leq \poly(n)
\]
time (this is the bottleneck of the verifier).
In Case 2, computing $s(\lat_0)$ using \cref{cor:comp-shortest-char-vec} takes $4^{n_0 + o(n_0)} = \poly(n)$ time, and the final two checks using the coefficient vector $\vec{z}$ take $\poly(n)$ time.
Therefore, the verifier takes $\poly(n)$ time overall, as needed.
\end{proof}

Put together, \cref{lem:SDLIP_K-NP,lem:SDLIP_K-coNP} immediately imply our main theorem about $\SDLIP_K$ being in $\NP \cap \coNP$, which we restate.

\latNPcoNP*

We also get the following corollary from the bound on $\NL(3)$ in \cref{thm:Nkbounds} and from \cref{hyp:finite_NLK}.

\begin{corollary}
$\SDLIP_3 \in \NP \cap \coNP$, and, if $\NL(K) < \infty$ for all $K \in \Z_{\geq 0}$ (i.e., if \cref{hyp:finite_NLK} holds), then $\SDLIP_K \in \NP \cap \coNP$ for any constant $K$.
\end{corollary}

%% file: SDPCE.tex
\section{On Deciding Equivalence of Self-Dual Codes}
\label{sec:SDPCE}

We now turn to studying Permutation Code Equivalence ($\PCE$) on self-dual codes.
The following lemma says that there is an efficient algorithm to compute a decomposition of a self-dual code $\C$ into a permutation of $\C_0 \oplus \cZ^r$.\footnote{We note in passing that one could also use Gaussian elimination to compute a generator matrix $G$ where $G^T$ is in reduced row echelon form to compute a decomposition of a direct-sum code $\C \cong \C_1 \oplus \cdots \oplus \C_m$ for some $m$ when the codes $\C_i$ for $i \geq 2$ each have constant length.}
The proof is analogous to the proof of \cref{thm:computing-reduced-lattice}. It uses the observation that finding the decomposition reduces to finding the weight-$2$ codewords in $\C$.

\begin{lemma}[{Implicit in~\cite{Elkies-long-shadows}}] \label{lem:computing-reduced-code}
There is a polynomial-time algorithm that, on input a generator matrix $G \in \F_2^{n \times n/2}$ of a self-dual code $\C \subset \F_2^n$, outputs a permutation matrix $P \in \cP_n$ such that $P(\C) = \C_0 \oplus \cZ^r$ for some $r \in \Z_{\geq 0}$, where $\C_0$ is a reduction of $\C$. The algorithm runs in $\poly(n)$ time.
\end{lemma}

\begin{proof}
The algorithm first computes the set $S := \set{\vec{c} \in \C : \wt(\vec{c}) = 2}$ of all weight-$2$ vectors in $\C$ by enumerating all $\binom{n}{2}$ vectors $\vec{c} \in \F_2^n$ with $\wt(\vec{c}) = 2$ and checking if $\vec{c} \in \C$. Let $r := \card{S}$, and let $\vec{c}_1, \ldots, \vec{c}_r$ be the elements of $S$ enumerated in some order. The algorithm outputs a permutation matrix $P$ corresponding to a permutation that maps $[n] \setminus (\union_{i = 1}^r \supp(\vec{c}_i))$ to $[n - 2r]$, and maps $\supp(\vec{c}_i)$ to $\set{n-2(r-i)-1, n-2(r-i)}$ for $i \in [r]$.

We first analyze the algorithm's running time. There are $\binom{n}{2} = O(n^2)$ weight-$2$ vectors $\vec{c} \in \F_2^n$, and checking whether a given vector $\vec{c}$ is in $\C = \C(G)$ takes $\poly(n)$ time.\footnote{If $G$ is in systematic form, checking whether $\vec{c} \in \C$ takes $O(n^2)$ time.}
Furthermore, computing $P$ given $S$ takes $O(n^2)$ time. So, the algorithm runs in $\poly(n)$ time.
Correctness follows by noting, as in~\cite{Elkies-long-shadows}, that any pair of codewords in $\C$ must either (1) have disjoint supports, or (2) be such that the support of one codeword is contained in the support of the other. In particular, the supports of the vectors in $S$ must be disjoint, and the support of $P\vec{c}_i$ for any vector $\vec{c}_i \in S$ must be disjoint from the support of $\C_0$.
\end{proof}

We then get our main theorem about search $\PCE$ on self-dual codes.

\lowRcodealg*

\begin{proof}
The algorithm works as follows. First, it uses the algorithm from \cref{lem:computing-reduced-code} to compute permutation matrices $P, P' \in \cP_n$ such that $P(\C) = \C_0 \oplus \cZ^r$, $P'(\C') = \C_0' \oplus \cZ^r$ for reduced codes $\C_0, \C_0'$.
It then uses either deterministic algorithm in~\cite{conf/soda/BabaiCGQ11} to solve search $\PCE$ on $\C_0, \C_0'$, obtaining $P_0 \in \cP_{n_0}$ such that $P_0(\C_0) = \C_0'$. Finally, it outputs $P_f := (P')^{-1} (P_0 \oplus I_{2r}) P \in \cP_n$.

By the assumption that $\C \permequiv \C'$ and the guarantees of \cref{lem:computing-reduced-code},~\cite{conf/soda/BabaiCGQ11}, $P_f(\C) = \C'$.
Furthermore, \cref{lem:computing-reduced-code} runs in $\poly(n)$ time on (generator matrices of) each of $\C$ and $\C'$, and the $\PCE$ algorithm in~\cite{conf/soda/BabaiCGQ11} runs in deterministic $2^{n_0 + o(n_0)}$ time on generator matrices of $\C_0, \C_0'$. The running time bound follows.
\end{proof}

We then obtain the following corollary from the fact that $\NC(K) \leq \NL(K)$ (\cref{fact:Nc<Nl}), the bounds on $\NL(K)$ in \cref{thm:Nkbounds}, and \cref{hyp:finite_NCK}. %
\begin{corollary}
For any $K = K(n)$ such that $\NC(K) = O(\log n)$, there is a polynomial-time algorithm for $\SearchSDPCE_K$.
In particular, there is a polynomial-time algorithm for $\SearchSDPCE_3$, and, if $\NC(K) < \infty$ for all $K \in \Z_{\geq 0}$ (i.e., if \cref{hyp:finite_NCK} holds), then there is a polynomial-time algorithm for $\SearchSDPCE_K$ for any constant $K$.
\end{corollary}

%% file: main.bbl
\newcommand{\etalchar}[1]{$^{#1}$}
\begin{thebibliography}{DPPvW22}

\bibitem[ADRS15]{conf/stoc/AggarwalDRS15}
Divesh Aggarwal, Daniel Dadush, Oded Regev, and Noah Stephens{-}Davidowitz.
\newblock Solving the shortest vector problem in $2^n$ time using discrete gaussian sampling.
\newblock In {\em STOC}, 2015.

\bibitem[ADS15]{conf/focs/AggarwalDS15}
Divesh Aggarwal, Daniel Dadush, and Noah Stephens{-}Davidowitz.
\newblock Solving the closest vector problem in $2^n$ time---the discrete gaussian strikes again!
\newblock In {\em FOCS}, 2015.

\bibitem[AKK{\etalchar{+}}25]{An2025ShortestSOEmbeddings}
Junmin An, Nathan Kaplan, Jon-Lark Kim, Jinquan Luo, and Guodong Wang.
\newblock Shortest self-orthogonal embeddings of binary linear codes, 2025.
\newblock Available at \url{https://arxiv.org/abs/2511.05440}.

\bibitem[APvW25]{conf/eurocrypt/AllombertPW25}
Bill Allombert, Alice Pellet{-}Mary, and Wessel P.~J. van Woerden.
\newblock Cryptanalysis of rank-2 module-lip: {A} single real embedding is all it takes.
\newblock In {\em {EUROCRYPT}}, 2025.

\bibitem[BBB{\etalchar{+}}26]{journals/tit/BennettBBDLLW26}
Huck Bennett, Drisana Bhatia, Jean{-}Fran{\c{c}}ois Biasse, Medha Durisheti, Lucas LaBuff, Vincenzo~Pallozzi Lavorante, and Phillip Waitkevich.
\newblock Asymptotic improvements to provable algorithms for the code equivalence problem.
\newblock {\em {IEEE} Transactions on Information Theory}, 72(2):1093--1108, 2026.

\bibitem[BBPS21]{conf/pqcrypto/BarenghiBPS21}
Alessandro Barenghi, Jean{-}Fran{\c{c}}ois Biasse, Edoardo Persichetti, and Paolo Santini.
\newblock {LESS-FM:} fine-tuning signatures from the code equivalence problem.
\newblock In {\em PQCrypto}, 2021.

\bibitem[BCGQ11]{conf/soda/BabaiCGQ11}
L{\'{a}}szl{\'{o}} Babai, Paolo Codenotti, Joshua~A. Grochow, and Youming Qiao.
\newblock Code equivalence and group isomorphism.
\newblock In {\em SODA}, 2011.

\bibitem[BGPS23]{conf/eurocrypt/BennettGPS23}
Huck Bennett, Atul Ganju, Pura Peetathawatchai, and Noah Stephens{-}Davidowitz.
\newblock Just how hard are rotations of $\mathbb{Z}^n$? {A}lgorithms and cryptography with the simplest lattice.
\newblock In {\em EUROCRYPT}, 2023.

\bibitem[BLS26]{BLS26-IBE-FHE-crypto}
Huck Bennett, Zhengnan Lai, and Noah Stephens{-}Davidowitz.
\newblock Advanced cryptography from lattice isomorphism—new constructions of {IBE} and {FHE}.
\newblock In {\em CRYPTO}, 2026.

\bibitem[BMM25]{branco25FHELIP}
Pedro Branco, Giulio Malavolta, and Zayd Maradni.
\newblock Fully-homomorphic encryption from lattice isomorphism.
\newblock In {\em TCC}, 2025.

\bibitem[BMPS20]{conf/africacrypt/BiasseMPS20}
Jean{-}Fran{\c{c}}ois Biasse, Giacomo Micheli, Edoardo Persichetti, and Paolo Santini.
\newblock {LESS} is more: Code-based signatures without syndromes.
\newblock In {\em AFRICACRYPT}, 2020.

\bibitem[BOS19]{conf/isit/BardetOS19}
Magali Bardet, Ayoub Otmani, and Mohamed Saeed{-}Taha.
\newblock Permutation code equivalence is not harder than graph isomorphism when hulls are trivial.
\newblock In {\em ISIT}, 2019.

\bibitem[CGG17]{chandrasekaranDecidingOrthogonalityConstructionA2017}
Karthekeyan Chandrasekaran, Venkata Gandikota, and Elena Grigorescu.
\newblock Deciding orthogonality in {{Construction}}-{{A}} lattices.
\newblock {\em SIAM Journal on Discrete Mathematics}, 31(2):1244--1262, January 2017.

\bibitem[CKM{\etalchar{+}}17]{Cohn2017-SpherePacking-Lambda24}
Henry Cohn, Abhinav Kumar, Stephen Miller, Danylo Radchenko, and Maryna Viazovska.
\newblock The sphere packing problem in dimension $24$.
\newblock {\em Annals of Mathematics}, 185(3), May 2017.

\bibitem[CS98]{Conway-Sloane-SPLAG}
John~H. Conway and Neil J.~A. Sloane.
\newblock {\em Sphere Packings, Lattices and Groups}.
\newblock Springer-Verlag, New York, 3 edition, 1998.

\bibitem[DG23]{conf/pkc/DucasG23}
L{\'{e}}o Ducas and Shane Gibbons.
\newblock Hull attacks on the lattice isomorphism problem.
\newblock In {\em PKC}, 2023.

\bibitem[Don83]{Donaldson1983}
S.~K. Donaldson.
\newblock An application of gauge theory to the topology of 4-manifolds.
\newblock {\em Journal of Differential Geometry}, 18:279--315, 1983.

\bibitem[DPPvW22]{conf/asiacrypt/DucasPPW22}
L{\'{e}}o Ducas, Eamonn~W. Postlethwaite, Ludo~N. Pulles, and Wessel P.~J. van Woerden.
\newblock Hawk: Module {LIP} makes lattice signatures fast, compact and simple.
\newblock In {\em ASIACRYPT}, 2022.

\bibitem[Duc24]{journals/dcc/Ducas24}
L{\'{e}}o Ducas.
\newblock Provable lattice reduction of $\mathbb{Z}^n$ with blocksize $n/2$.
\newblock {\em Designs, Codes and Cryptography}, 92(4):909--916, 2024.

\bibitem[DvW22]{conf/eurocrypt/DucasW22}
L{\'{e}}o Ducas and Wessel P.~J. van Woerden.
\newblock On the lattice isomorphism problem, quadratic forms, remarkable lattices, and cryptography.
\newblock In {\em EUROCRYPT}, 2022.

\bibitem[Elk95a]{Elkies-char-Zn}
Noam~D. Elkies.
\newblock A characterization of the $\mathbb{Z}^n$ lattice.
\newblock {\em Mathematical Research Letters}, 2:321–326, 1995.

\bibitem[Elk95b]{Elkies-long-shadows}
Noam~D. Elkies.
\newblock Lattices and codes with long shadows.
\newblock {\em Mathematical Research Letters}, 2:643--651, 1995.

\bibitem[Fre82]{Freedman-four-manifolds-82}
Michael~Hartley Freedman.
\newblock {The topology of four-dimensional manifolds}.
\newblock {\em Journal of Differential Geometry}, 17(3):357 -- 453, 1982.

\bibitem[Gab07]{Gaborit2007}
Philippe Gaborit.
\newblock A bound for certain $s$-extremal lattices and codes.
\newblock {\em Archiv der Mathematik}, 89(2):143–151, Apr 2007.

\bibitem[Gau98a]{Gaulter-Lattices-Without-Short}
Mark Gaulter.
\newblock Lattices without short characteristic vectors.
\newblock {\em Mathematical Research Letters}, 5(3):353--362, 1998.

\bibitem[Gau98b]{gaulter1998characteristic}
Mark~James Gaulter.
\newblock {\em Characteristic Vectors of Unimodular Lattices over the Integers}.
\newblock {Ph.D.} thesis, University of California, Santa Barbara, 1998.

\bibitem[Gau07]{Gaulter-S2}
Mark Gaulter.
\newblock Characteristic vectors of unimodular lattices which represent two.
\newblock {\em Journal de Théorie des Nombres de Bordeaux}, 19(2):405--414, 2007.

\bibitem[Ger04]{Gerstein04}
Larry~J. Gerstein.
\newblock Characteristic elements of unimodular $\mathbb{Z}$-lattices.
\newblock {\em Linear and Multilinear Algebra}, 52(5):381--383, 2004.

\bibitem[HR14]{havivLatticeIsomorphismProblem2014}
Ishay Haviv and Oded Regev.
\newblock On the {{Lattice Isomorphism Problem}}.
\newblock In {\em {{SODA}}}, 2014.

\bibitem[Hun24]{Hunkenschroder-ZLIP-coNP}
Christoph Hunkenschr{\"o}der.
\newblock Deciding whether a lattice has an orthonormal basis is in {co-NP}.
\newblock {\em Mathematical Programming}, 208(1--2):763--775, 2024.

\bibitem[LRvW25]{leporati2025beyond}
Alberto Leporati, Lorenzo Rovida, and Wessel van Woerden.
\newblock Beyond {LWE}: a lattice framework for homomorphic encryption.
\newblock {\em Cryptology ePrint Archive}, 2025.

\bibitem[LS17]{lenstraLatticesSymmetry2017}
Hendrik Lenstra and Alice Silverberg.
\newblock Lattices with symmetry.
\newblock {\em Journal of Cryptology}, 30(3):760--804, 2017.

\bibitem[MPPW24]{conf/eurocrypt/MureauPPW24}
Guilhem Mureau, Alice Pellet{-}Mary, Georgii Pliatsok, and Alexandre Wallet.
\newblock Cryptanalysis of rank-2 module-lip in totally real number fields.
\newblock In {\em EUROCRYPT}, 2024.

\bibitem[MV13]{journals/siamcomp/MicciancioV13}
Daniele Micciancio and Panagiotis Voulgaris.
\newblock A deterministic single exponential time algorithm for most lattice problems based on voronoi cell computations.
\newblock {\em {SIAM} Journal on Computing}, 42(3):1364--1391, 2013.

\bibitem[{Nat}25]{nist_pqc_round2_additional_signatures_2025}
{National Institute of Standards and Technology (NIST)}.
\newblock {Post-Quantum Cryptography: Additional Digital Signature Schemes --- Round 2 Additional Signatures}.
\newblock \url{https://csrc.nist.gov/projects/pqc-dig-sig/round-2-additional-signatures}, 2025.

\bibitem[NS07]{NebeSchindelar2007}
Gabriele Nebe and Kristina Schindelar.
\newblock S-extremal strongly modular lattices.
\newblock {\em Journal de Théorie des Nombres de Bordeaux}, 19(3):683--701, 2007.

\bibitem[NV03]{NebeVenkov2003}
Gabriele Nebe and Boris Venkov.
\newblock Unimodular lattices with long shadow.
\newblock {\em Journal of Number Theory}, 99(2):307--317, 2003.

\bibitem[Ple72]{pless-selforthcodes-1972}
Vera Pless.
\newblock A classification of self-orthogonal codes over {GF(2)}.
\newblock {\em Discrete Mathematics}, 3(1):209--246, 1972.

\bibitem[Sen00]{journals/tit/Sendrier00}
Nicolas Sendrier.
\newblock Finding the permutation between equivalent linear codes: The support splitting algorithm.
\newblock {\em {IEEE} Transactions on Information Theory}, 46(4):1193--1203, 2000.

\bibitem[SSV09]{journals/moc/SikiricSV09}
Mathieu~Dutour Sikiric, Achill Sch{\"{u}}rmann, and Frank Vallentin.
\newblock Complexity and algorithms for computing {Voronoi} cells of lattices.
\newblock {\em Mathematics of Computation}, 78(267):1713--1731, 2009.

\bibitem[Ste26]{PersComm-SDlat-StephensD26}
Noah Stephens{-}Davidowitz.
\newblock Personal communication, 2026.

\bibitem[Szy03]{szydloHypercubicLatticeReduction2003}
Michael Szydlo.
\newblock Hypercubic lattice reduction and analysis of {{GGH}} and {{NTRU}} signatures.
\newblock In {\em {{EUROCRYPT}}}, 2003.

\bibitem[Via17]{Viazovska2017-SpherePacking-E8}
Maryna Viazovska.
\newblock The sphere packing problem in dimension $8$.
\newblock {\em Annals of Mathematics}, 185(3), May 2017.

\bibitem[vW24]{vanWoerden-genus-asisacrypt}
Wessel van Woerden.
\newblock Dense and smooth lattices in any genus.
\newblock In {\em Asiacrypt}, 2024.

\end{thebibliography}
